%% file: main.tex
\documentclass[11pt]{article}

\newif\ifcomments

\title{Dynamic Load Balancing under Restricting Updates}

\author{}

\author{Aaron Bernstein\thanks{
        New York University,
        \texttt{bernstei@gmail.com}. Supported by Sloan Fellowship, NSF Grant 1942010, and Charles S. Baylis endowment at NYU.
    } \and Sayan Bhattacharya\thanks{
        University of Warwick,
        \texttt{S.Bhattacharya@warwick.ac.uk}.
    }\and Peter Kiss\thanks{
        University of Vienna,
        \texttt{peter.kiss@univie.ac.at} This research was funded in whole or in part by the Austrian Science Fund (FWF) 10.55776/ESP6088024, Part of this work was done while at the University of Warwick.
    }  \and Thatchaphol Saranurak\thanks{
        University of Michigan,
        \texttt{thsa@umich.edu}.
        Supported by NSF Grant CCF-2238138 and a Sloan Fellowship.
        Part of this work was done while at INSAIT, Sofia University ``St. Kliment Ohridski'', Bulgaria. This work was partially funded from the Ministry of Education and Science of Bulgaria (support for INSAIT, part of the Bulgarian National Roadmap for Research Infrastructure).
    } }

\date{}

\input{commands}
\begin{document}

\maketitle
\pagenumbering{gobble}

\begin{abstract}
   \input{abstract.tex} 
\end{abstract}

\newpage
\tableofcontents
\newpage

\pagenumbering{arabic}
\input{intro}

\input{overview}

\input{prelim}

\input{fractional-algorithm}

\input{Rounding}

\input{Acknowledgements.tex}

\bibliographystyle{alpha}
\bibliography{references}

\appendix

\input{appendix-hardness}

\end{document}

%% file: commands.tex
\usepackage[linesnumbered,ruled,vlined]{algorithm2e}

\usepackage{fancyhdr}
\usepackage{blindtext}
\usepackage{amsthm}
\usepackage{amsmath}
\usepackage{amssymb}
\usepackage{thmtools}
\usepackage{thm-restate}
\usepackage{anyfontsize}
\usepackage{xstring}
\usepackage{xargs}
\usepackage{makeidx}
\usepackage{cancel}
\usepackage{chngcntr}
\usepackage{caption}
\usepackage{subcaption}
\usepackage{float}
\usepackage{pgf,tikz,pgfplots}
\usepackage{latexsym}
\usepackage{verbatim}
\usepackage{bookmark}
\usepackage[toc]{appendix}
\usepackage[a4paper,left=2cm,right=2cm,top=3cm,bottom=2.4cm]{geometry}

\usepackage{cleveref}

\usepackage{listings}
\usepackage{hyperref}

\hypersetup{
	colorlinks=true,
	linkcolor=blue,
	filecolor=magenta,      
	urlcolor=cyan,
}

\usepackage{algorithmicx}
\usepackage{algpseudocode}
\usepackage[dvipsnames]{xcolor}

\usepackage{bidicontour}

\newtheorem{theorem}{Theorem}[section]
\newtheorem{lemma}[theorem]{Lemma}

\newtheorem{corollary}[theorem]{Corollary}
\newtheorem{proposition}[theorem]{Proposition}

\newtheorem{definition}[theorem]{Definition}

\newtheorem{observation}[theorem]{Observation}
\newtheorem{claim}[theorem]{Claim}

\hypersetup{ colorlinks, citecolor=blue}

\PassOptionsToPackage{unicode}{hyperref}

\usepackage{framed}
\usepackage[framemethod=tikz]{mdframed}
\usepackage{tikz-cd}
\newenvironment{wrapper}[1]
{
	\begin{center}
		\begin{minipage}{\linewidth}
			\begin{mdframed}[hidealllines=true, backgroundcolor=gray!20, leftmargin=0cm,innerleftmargin=0.4cm,innerrightmargin=0.4cm,innertopmargin=0.4cm,innerbottommargin=0.4cm,roundcorner=0pt]
				#1}
			{\end{mdframed}
		\end{minipage}
	\end{center}
}

\ifcomments

\newcommand{\thatchaphol}[1]{{\color{purple}Thatchaphol: #1}}

\newcommand{\sayan}[1]{{\color{red}Sayan: #1}}

\newcommand{\aaron}[1]{{\color{magenta}Aaron: #1}}

\newcommand{\peter}[1]{{\color{orange}Peter: #1}}

\newcommand{\red}[1]{{\color{red}#1}}

\newcommand{\TODO}[1]{\color{red} TODO: #1 \color{black}}

\else 
\newcommand{\thatchaphol}[1]{}

\newcommand{\sayan}[1]{}

\newcommand{\aaron}[1]{}

\newcommand{\peter}[1]{}

\newcommand{\red}[1]{}

\newcommand{\TODO}[1]{}
\fi

\newcommand{\opt}{\textsc{Opt}}
\newcommand{\boost}{\texttt{Boost}}
\newcommand{\whack}{\texttt{Whack}}
\newcommand{\preprocess}{\texttt{Preprocess}}
\newcommand{\delete}{\texttt{Handle-Update}}
\newcommand{\resetbucket}{\texttt{Reset-Bucket}}
\newcommand{\increaseslot}{\texttt{Increase-Weight}}
\newcommand{\decreaseslot}{\texttt{Decrease-Weight}}
\newcommand{\G}{\mathcal{G}}
\newcommand{\J}{\mathcal{J}}
\newcommand{\M}{\mathcal{M}}
\newcommand{\E}{\mathcal{E}}

\DeclareMathOperator{\poly}{poly}

\newcommand{\gain}{\texttt{gain}}

\newcommand{\loss}{\texttt{loss}}

\usepackage{ifthen}
\newboolean{debug}
\setboolean{debug}{false} 

\DeclareMathOperator{\polylog}{polylog}
\newcommand{\eps}{\epsilon}
\newcommand{\otil}{\tilde{O}}

%% file: abstract.tex
We study the problem of load balancing in the dynamic setting. The input is a bipartite graph $G = (J \cup M, E)$ of jobs $J$ and machines $M$. Every job has a positive weight $w(j)$. The goal is to assign every job in $J$ to an adjacent machine in $M$ so as to minimize the maximum machine load, where the load of a machine is the sum of the weights of the jobs assigned to it.

Although load balancing has been studied extensively in many models of computation, the only previous result in the dynamic model addresses the \emph{incremental} setting, where updates take the form of edge insertions and job deletions \cite{BhattacharyaKS23}. The authors show a $(1+\eps)$-approximation for incremental \emph{fractional} load balancing with $\poly(\log(n),1/\eps)$ update time; they point out that their techniques do not extend to the decremental setting. 

\paragraph{Main Result:} We obtain a near-optimal algorithm for fractional load balancing in the \emph{decremental} setting, where updates take the form of edge deletions and job insertions: a $(1+\eps)$-approximation with $\poly(\log(n),1/\eps)$ update time. Our techniques are very different from those of the incremental algorithm of \cite{BhattacharyaKS23}. Instead, our algorithm is inspired by a load-balancing algorithm from the {\sc congest} model of computation \cite{Ahmadian2021DistributedLB}. We develop a novel analytic framework for this algorithm, which allows us to show that it performs well in the dynamic setting.

\paragraph{Rounding Scheme:} We also present a black-box rounding scheme for dynamic load balancing. Plugging in our fractional algorithm, we obtain a decremental algorithm for $(2+\eps)$-approximate \emph{integral} load balancing with $\poly(\log(n),1/\eps)$ update time. We also obtain the same bounds for incremental load balancing by plugging in the fractional algorithm of \cite{BhattacharyaKS23}. The approximation ratio nearly matches the $2$-approximation achieved by the state-of-the-art \emph{static} algorithms for this problem. 


%% file: intro.tex
\section{Introduction}

We study the load balancing problem in the dynamic setting. The input is a bipartite graph $G = (J \cup M, E)$, where we refer to the vertices in $J$ as jobs, and those in $M$ as machines. Each job also has a positive weight $w(j)$. The goal is to find an assignment from jobs to adjacent machines that minimizes the maximum machine load, where the load of a machine $i$ is defined as the sum of the weights of all jobs assigned to $i$.

There is a rich literature on load balancing in many models of computation, from standard approximation algorithms \cite{LenstraST90,ShmoysT93,Svensson11,JansenR20}, to classic online algorithms \cite{AzarNR95,AspnesAFPW97}, streaming algorithms \cite{KonradR16,AssadiBL23,AssadiBLLW25}, and parallel/distributed algorithms \cite{Young2001MixedPackingCovering,
  Mahoney2016ParallelMixedPackingCovering, Li23,HalldorssonKPR18,AssadiBL20,Ahmadian2021DistributedLB}. One particularly relevant line of research is online algorithms with recourse, where the graph is changing over time, and the goal is to upper-bound how often the (approximately) optimal assignment needs to change \cite{GuptaKS14, BerndtEM22, Krishnaswamy0S23, FoussoulGK24}. But the results for this model do not provide efficient algorithms for actually maintaining the assignment -- that is the goal of dynamic algorithms.

\paragraph{Dynamic Load Balancing:} In the \emph{fully dynamic} model, each update either inserts or deletes an edge in the graph, and the goal is to maintain an approximately optimal assignment with small update time. There are also two commonly studied special cases, which are called \emph{partially dynamic} models: in the incremental (resp. decremental) model the only allowed updates are edge insertions (resp.~deletions) and job deletions (resp. insertions); in short, updates in the incremental model can only decrease the optimum max load, while those in the decremental model can only increase it. 

Despite the rich literature on load balancing in other models, the only result we know of on dynamic load balancing is an algorithm of \cite{BhattacharyaKS23} for incremental load balancing: they maintain a $(1+\eps)$-approximate fractional assignment in \emph{total} update time $\otil(m)$, which implies amortized update time $\poly(\log(n)/\eps)$.
However, the authors of \cite{BhattacharyaKS23} note that their techniques fail to extend to the decremental setting.

Our main result obtains the same bounds in the decremental setting. We also develop a black-box rounding scheme to produce an integral assignment. 
We thus essentially settle the complexity of the problem for partially dynamic algorithms. One of the reasons to focus on the partially dynamic model is that similarly strong results for the fully dynamic model are currently beyond our reach, even for the much simpler problem of dynamic matching; see the discussion in \Cref{sec:related} on the relationship to matching.

\subsection{Our Contributions}

Our main contribution is a nearly optimal algorithm for \emph{fractional} load balancing in the decremental model.

\begin{theorem}
\label{thm:intro-fractional} Let $G = (J \cup M, E)$ be a graph with job weights in $[1,W]$ that is subject to a sequence of edge deletions and job insertions. There is an algorithm that maintains a \emph{fractional} assignment whose maximum machine load is a $(1+\eps)$-approximation to the optimum. The total update time is $O(m\poly(\log(nW)/\eps))$, where $m$ is the total number of edges ever present in the graph and $n = |J \cup M|$ is the total number of nodes at the end of the last update.
\end{theorem}

 We also design a black-box rounding scheme, which is based on an earlier rounding scheme of \cite{Krishnaswamy0S23} for online algorithms, but makes several significant modifications to ensure efficient update time.
Applying the rounding scheme (\Cref{sec:rounding}) to Theorem \ref{thm:intro-fractional} yields the following integral algorithm:

\begin{theorem}
\label{thm:intro-integral} Let $G = (J \cup M, E)$ be a graph with job weights in $[1,W]$ that is subject to a sequence of edge deletions and job insertions. There is a randomized algorithm that maintains an \emph{integral} assignment whose maximum machine load is a $(2+\eps)$-approximation to the optimum, using total update time $O(m\poly(\log(nW)/\eps))$ with high probability. In fact, the approximation guarantee is slightly stronger: each machine in our assignment has load at most $(1+\eps)\cdot \opt + W \leq (2+\eps) \cdot \opt$. 
The algorithm works against an adaptive adversary.
\end{theorem}

The approximation ratio of the integral algorithm is essentially the best we can hope for: it matches up to a $(1+\eps)$ factor the best-known polynomial-time algorithms in the \emph{static} setting~\cite{LenstraST90,ShmoysT93}.\footnote{There are better-than-two approximations for the \emph{value} version of the problem, but these cannot compute an actual assignment; they also have large polynomial running time and hence cannot lead to efficient dynamic algorithms~\cite{Svensson11,JansenR20}.} 

Applying our black-box rounding scheme to the incremental (fractional) algorithm of \cite{BhattacharyaKS23} also yields the first efficient algorithm for incremental \emph{integral} load balancing; the bounds are the same as in Theorem \ref{thm:intro-integral}.

\subsection{Related Work: Maximum Flow and Maximum Matching}
\label{sec:related}

Our results have significant connections to dynamic algorithms for several related problems.

\paragraph{Maximum Flow:} It is not hard to check that incremental/decremental load balancing is a special case of incremental/decremental maximum flow. There is an algorithm for incremental $(1+\eps)$-approximate maximum flow with update time $n^{o(1)}$ \cite{ChenKLMG24}. This same algorithm thus also solves incremental load balancing, although the algorithm is significantly more complicated than the incremental load balancing algorithm of \cite{BhattacharyaKS23} and has update time $n^{o(1)}$ rather than $\polylog$. In the decremental setting, by contrast, while there exists a recent flow algorithm with the same $n^{o(1)}$ update time \cite{BrandCKLMGS24}, \emph{that algorithm cannot actually maintain the approximate assignment itself}; it only maintains a single number, which is an approximation to the optimum max load. Ours is thus the first efficient algorithm to maintain the assignment itself in the decremental setting.

\paragraph{Matching:} Load balancing can be viewed as a generalization of the matching problem. There is a long line of work on this problem in partially dynamic settings, and here too efficient incremental algorithms \cite{Gupta2014,GrandoniLSSS2019,BlikstadKiss2023} came before decremental ones \cite{BernsteinGS2020,JambulapatiJST2022,AssadiBD2022,Dudeja2024,ChenST2025}. 


The connection to matching also implies a crucial barrier for our problem. Even for the simpler problem of matching, a \emph{fully dynamic} $(1+\eps)$-approximation in $\polylog$ (or even $n^{1-\Omega(1)}$) update time is not known, and this is widely considered one of the main open problems in the field. In fact, there are conditional lower bounds that rule out such a result if we seek a polynomial dependence on $1/\eps$ \cite{Liu24}. At the same time, it is not hard to show that a $(1+\eps)$-approximation for \emph{fractional} load balancing would immediately imply the same bounds for fully dynamic \emph{integral} matching; see Appendix~\ref{ap:matching} for details. Generalizing our $(1+\eps)$-approximate decremental load balancing algorithm to the fully dynamic setting is thus out of reach, as we should first aim to understand the simpler problem of fully dynamic $(1+\eps)$-approximate matching.

That still leaves open the question of whether one can get an efficient algorithm for fully dynamic load balancing with a worse approximation. We show a conditional lower bound which proves that even in unweighted graphs, a $2-\Omega(1)$-approximation for \emph{integral} load balancing requires large update time (Appendix~\ref{app:integralhardness}). But neither upper nor lower bounds are known for approximation better than $2$. Matching is no longer a barrier here, as an efficient algorithm for fully dynamic $2$-approximate matching is known \cite{BaswanaGS11,Solomon16}. At the same time, load balancing is a significantly harder problem than matching, and in fact there are models of computation where it is provably \emph{much} harder \cite{AssadiBLLW25}. We leave it as an open question whether there exists an efficient algorithm for $O(1)$- or even $\polylog$-approximate fully dynamic load balancing, or whether there are conditional lower bounds suggesting such bounds are not possible.

\paragraph{Mixed packing-covering LP} \cite{BhattacharyaKS23} showed how to maintain a $(1+\eps)$-approximate solution of mixed packing-covering LPs under relaxing (i.e.~incremental) updates in near-optimal update time, and left as an open problem whether non-trivial algorithms exist for ``restricting'' (i.e.~decremental) updates. Load balancing is one of the most basic instances of mixed packing-covering LP, and we show a near-optimal algorithm for it under restricting updates.  
Hence, this work makes a first step toward resolving the open problem of \cite{BhattacharyaKS23}. 


%% file: overview.tex
\section{Overview}

\label{sec:overview}

We give an informal overview of the two main ingredients of the paper. First, we explain a dynamic algorithm for maintaining a nearly optimal \emph{fractional} assignment in \Cref{sec:overview:frac}. Second, we explain in \Cref{sec:overview:round} a black-box dynamic \emph{rounding} algorithm that converts such a fractional assignment into an integral assignment with only a small additional loss.

\subsection{Maintaining a Fractional Assignment}

\label{sec:overview:frac}

Let $G=(J,M,E)$ be the bipartite graph between jobs $J$ and machines $M$. The graph undergoes edge deletions and job insertions. Let $n=|J|+|M|$ and let $m$ be the number of edges that ever appear in $G$.

To keep the exposition clean, we describe the fractional algorithm in the simple case where all jobs have weight $1$ and where we assume that at all times the current graph $G$ admits an \emph{integral} assignment $x^{*}\in\{0,1\}^{E}$ where each job sends $x^{*}(j):=\sum_{i\in M}x^{*}(j,i)=1$ unit of flow and each machine receives $x^{*}(i):=\sum_{j\in J}x^{*}(j,i)\le1$ unit. Our goal is to maintain a \emph{fractional} assignment $x\in[0,1]^{E}$ such that each job sends $x(j):=\sum_{i}x(j,i)\ge1/(1+\eps)$ units and each machine receives $x(i):=\sum_{j}x(j,i)\le(1+\eps)\beta$, where $\beta:=1+10\eps.$ Scaling all values of $x$ by a factor of at most $(1+\eps)$ then gives a fractional assignment of load at most $(1+\eps)^{2}\beta=1+O(\eps)$. The assumption that $x^{*}$ is integral is only for the sake of this overview; the actual proof handles fractional optimum assignments by replacing the injection arguments below with flow-counting arguments.

\paragraph{Starting point: CONGEST Algorithm.}

The starting point is an extremely simple multiplicative-update algorithm from the \textsc{congest} literature by Ahmadian, Liu, Peng, and Zadimoghaddam~\cite{Ahmadian2021DistributedLB}. Each job $j$ and each machine $i$ has a nonnegative integer level, denoted by $\ell(j)$ and $\ell(i)$, initialized ats zero. The algorithm maintains a flow value on each current edge $(j,i)$ of the form 
\[
x(j,i)=\eta\cdot(1+\eps)^{\ell(j)-\ell(i)},
\]
where $\eta>0$ is a very small number; for the unweighted overview, one may take $\eta=n^{-20}$. Increasing the level of a job multiplies all of its incident flow values by $(1+\eps)$, while increasing the level of a machine divides all of its incident flow values by $(1+\eps)$. We call these two operations \emph{boosting} the job and \emph{whacking} the machine.

The algorithm repeatedly performs the following local repairs. If some machine has load larger than $(1+\eps)\beta$, we whack it. Otherwise, if some job sends less than $1/(1+\eps)$ total flow, we boost it. We always give priority to whacking machines before boosting jobs. This priority rule ensures that every machine has load at most $(1+\eps)^2\beta < 2\beta$ throughout the algorithm.

When the local repair loop stops, every job sends at least $1/(1+\eps)$ flow and every machine receives load at most $(1+\eps)\beta$, as desired. 
Thus, the only remaining issue is to prove that the repair loop cannot run for too long while $x^{*}$ exists. Equivalently, we need to prove the following key claim. 
\begin{quote}
\textbf{Key claim.} If the current instance admits an assignment of maximum load one, then no job or machine can reach level $L_{\max}=\poly(\log n,1/\eps)$. 
\end{quote}
Once this claim is known, the total running time is immediate. A boost or whack costs time proportional to the current degree of the vertex being updated. Since each edge is charged only when one of its endpoints increases level, and each endpoint increases level only $L_{\max}$ times, the total update time is $\tilde{O}(mL_{\max})$, near-linear in the number of edges that ever appear.

\paragraph{A warm-up without edge deletions.}

Ahmadian, Liu, Peng, and Zadimoghaddam~\cite{Ahmadian2021DistributedLB} proved the key claim when no edge is deleted using the following crucial observation: \emph{every machine $i$ with positive level receives flow at least $x(i)\ge\beta$}. Indeed, once $i$ was whacked, its load right after the whack was at least $\beta$ and can only increase until $i$ is whacked again. On the other hand, every job always sends total flow $x(j)<1$.

Let $M_{\ge k}$ and $J_{\ge k}$ be the machines and jobs whose current levels are at least $k$. The proof has two simple inequalities. First, for every $k>0$, 
\begin{equation}
|J_{\ge k}|\ge\beta|M_{\ge k}|-\text{(negligible term)}.\label{eq:warmup-job-expansion}
\end{equation}
Consider the total flow from $J_{\ge k}$ to $M_{\ge k}$, denoted by $x(J_{\ge k},M_{\ge k})$. Observe that every machine $i$ essentially only receives flow from jobs $j$ with higher level $\ell(j)>\ell(i)$. Otherwise, if $\ell(j)\le\ell(i)$, the flow value $x(j,i)=\eta(1+\eps)^{\ell(j)-\ell(i)}\le\eta$ is negligible. Thus, $x(J_{\ge k},M_{\ge k})\ge|M_{\ge k}|\beta-\text{(negligible term)}$ as each machine receives at least $\beta$. On the other hand, $x(J_{\ge k},M_{\ge k})\le|J_{\ge k}|$ as every job sends less than one unit. \ref{eq:warmup-job-expansion} follows from combining the two bounds. 

Second, let $L := \lceil \log_{1+\epsilon}(1/\eta)\rceil = O(\log(n)/\epsilon)$. 
For every $k\ge0$,
\begin{equation}
|M_{\ge k-L}|\ge|J_{\ge k}|.\label{eq:warmup-machine-expansion}
\end{equation}
To see this, take a job $j\in J_{\ge k}$ and let $i$ be the machine to which $j$ is assigned by the assumed optimum assignment $x^{*}$. Since $(j,i)$ is a current edge and $1>x(j) \ge x(j,i)=\eta\cdot(1+\eps)^{\ell(j)-\ell(i)}$, we must have $\ell(i)\ge\ell(j)-L\ge k-L$. Thus $x^{*}$ assigns all jobs in $J_{\ge k}$ into only machines in $M_{\ge k-L}$.
Since $x^{*}$ has load at most one and jobs have unit size, this assignment is injective, proving~\eqref{eq:warmup-machine-expansion}.

Combining the two inequalities and ignoring the negligible term yields 
\[
|M_{\ge k-L}|\ge\beta|M_{\ge k}|.
\]
Thus, as we lower the level threshold, the number of high-level machines grows geometrically. This cannot happen for more than logarithmically many thresholds, so levels are bounded by $O(L\log_{\beta}n)=\poly(\log n,1/\eps)$.

This is the whole strategy: to forbid high-level vertices, we establish a \emph{job-machine flow gap} and apply an \emph{expansion} argument over level sets as above.

\paragraph{Why deletions break the warm-up.}
Under edge deletions, the sentence ``every machine $i$ with positive level receives flow at least $x(i)\ge\beta$'' is completely false. While the load of machine $i$ is still at least $\beta$ right after being whacked, it may decrease arbitrarily later because incident edges are deleted.

\paragraph{First key idea: history graph} To fix this, our first key idea is to introduce another graph that records the relevant events of the algorithm, including the moments when machine loads are large. We call it the \emph{history graph} $\G=(\J,\M,\E)$. For every job $j$ and every level $k$ that $j$ has ever reached, the history graph contains a \emph{job-level node} $(j,k)\in\J$. Similarly, for every machine $i$ and every level $k$ that $i$ has ever reached, it contains a \emph{machine-level node} $(i,k)\in\M$. Whenever an original edge $(j,i)$ is present while the endpoints have levels $k_{j}$ and $k_{i}$, the history graph contains an edge $((j,k_{j}),(i,k_{i}))\in\E$. The \emph{history flow} $z\in[0,1]^{\E}$ sends 
\[
z((j,k_{j}),(i,k_{i}))=\eta\cdot(1+\eps)^{k_{j}-k_{i}}
\]
units over this history edge. This is exactly the value of the original flow on $(j,i)$ at the time this historical event occurred. The key property is that the history graph is \emph{incremental} under edge deletions: even after an edge is deleted from the current graph, the corresponding historical record is kept, and its $z$-value never changes.

This idea alone already handles edge deletions, but only if we allow a $O(1/\eps)$-approximation to the optimum max load. Let $z(i,k)$ be the total incoming history flow of the machine-level node $(i,k)$, and let $z(j,k)$ be the total outgoing history flow of the job-level node $(j,k)$. Each machine-level node $(i,k)$ with $k>0$ receives at least 
\[
z(i,k)\ge\beta
\]
 units, because this was true right after the corresponding whack, and later history edges can only add to $z(i,k)$. On the other hand, a single job-level node $(j,k)$ may send as much as 
\[
z(j,k)\le1+\frac{1}{1+\eps}+\frac{1}{(1+\eps)^{2}}+\cdots\le\frac{1+\eps}{\eps}
\]
units of history flow. This happens because, while $j$ stays at level $k$, neighboring machines may be whacked many times. Each such whack reduces the current flow in the original graph by a $1+\eps$ factor, but it creates a new historical record and increases $z(j,k)$. If we set $\beta=2/\eps$, this creates the desired gap: $z(i,k)\ge2/\eps$ for all machine-level nodes $(i,k)$ and $z(j,k)\le\frac{1+\eps}{\eps}$ for all job-level nodes $(j,k)$, and we can apply a similar expansion argument to bound the level of vertices. However, as $\beta=\Theta(1/\eps)$, this gives only an $O(1/\eps)$-approximation.

\paragraph{Second key idea: average over suffixes of levels.}

The second key idea is to stop looking at one history node at a time. A single job-level node can be misleading, because it may accumulate a long geometric tail of historical records while the job waits at that level. Instead, for a vertex $v$ whose current level is $\ell(v)$, and for a threshold $k\le\ell(v)$, we look at the whole suffix of history nodes 
\[
(v,k),(v,k+1),\ldots,(v,\ell(v)).
\]
The point is that the bad $\Theta(1/\eps)$ contribution for a job is an additive cost over the whole suffix, not a cost paid at every level. Once the suffix is long, this additive cost is small compared to the number of levels in the suffix.

Let us make this precise for jobs. Fix a job $j$ and a threshold $k\le\ell(j)$, and write 
\[
q:=\ell(j)-k+1
\]
for the number of job-level nodes in the suffix. We want to upper bound 
\[
Z(j,k):=\sum_{r=k}^{\ell(j)}z(j,r),
\]
the total outgoing history flow of this suffix. To do this, we will decompose $Z(j,k)$ into two parts as follows. For a level $r\in[k,\ell(j)]$, we write $z(j,r)=z^{\text{init}}(j,r)+z^{\text{new}}(j,r)$ where $z^{\text{init}}(j,r)$ is the total flow out of $(j,r)$ right after $j$ is boosted into level $r$, and $z^{\text{new}}(j,r)$ is the increase in $z(j,r)$ while $j$ remains at some level $r$ and neighboring machines are whacked. Then, define $Z^{\text{init}}(j,k):=\sum_{r=k}^{\ell(j)}z^{\text{init}}(j,r)$ and $Z^{\text{new}}(j,k):=\sum_{r=k}^{\ell(j)}z^{\text{new}}(j,r)$. Now, we can decompose $Z(j,k)$ as
\[
Z(j,k)=Z^{\text{init}}(j,k)+Z^{\text{new}}(j,k).
\]
As $j$ sends less than $1$ unit right after being boosted, we have $z^{\text{init}}(j,r)<1$ for every $r$ and so 
\[
Z^{\text{init}}(j,k)<q.
\]
Next, we bound $Z^{\text{new}}(j,k)$. During the whole suffix from level $k$ to $\ell(j)$, let $\text{Loss}(j,k)$ and $\text{Gain}(j,k)$ be the total decrease and total increase of $x(j)$, respectively. That is, whenever the flow $x(j)$ decreases while $\ell(j) \geq k$ (for example, because a neighboring machine is whacked, or because an edge gets deleted), this decrease contributes to $\text{Loss}(j,k)$; $\text{Gain}(j,k)$ is defined analogously.

Observe that $Z^{\text{new}}(j,k)\le \frac{1}{\epsilon}\text{Loss}(j,k)$ because,
whenever neighboring machine $i$ is whacked, $x(j)$ decreases by $x(j,i)-\frac{1}{1+\eps}x(j,i)=\frac{\eps}{1+\eps}x(j,i)$ but $Z^{\text{new}}(j,k)$ increases by $\frac{1}{1+\eps}x(j,i)$. This inequality would actually be an equality without edge deletions, but edge deletions can contribute to $\text{Loss}(j,k)$ without increasing $Z^{\text{new}}(j,k)$.
Next, we have $\text{Loss}(j,k)\le\text{Gain}(j,k)+1$ because $x(j)$ always remains between zero and one. Furthermore, $\text{Gain}(j,k)\le\frac{\eps}{1+\eps}q$ because the maximum increase of $x(j)$ per level is $1-\frac{1}{1+\eps}=\frac{\epsilon}{1+\epsilon}$. Combining three bounds, we get
\[
Z^{\text{new}}(j,k)\le q+\frac{1}{\eps}
\]
and, hence, 
\begin{equation}
Z(j,k)\le2q+\frac{1}{\eps}.\label{eq:overview-job-suffix}
\end{equation}
In particular, if $q\ge\alpha$, where $\alpha:=\lceil100/\eps^{2}\rceil$, then 
\begin{equation}
Z(j,k)\le\left(2+\frac{\eps}{10}\right)q.\label{eq:overview-job-suffix-long}
\end{equation}
So a long job suffix sends, on average, only slightly more than two units of history flow per level.

Next, we prove a lower bound for the machine side using an analogous argument.
Consider a machine $i$ and a threshold $k\le\ell(i)$.
Let 
\[
q:=\ell(i)-k+1\qquad\text{and}\qquad Z(i,k):=\sum_{r=k}^{\ell(i)}z(i,r).
\]
We want a lower bound on $Z(i,k)$ by similarly decomposing it into two parts $Z(i,k)=Z^{\text{init}}(i,k)+Z^{\text{new}}(i,k).$ We have, for $k>0$, 
\[
Z^{\text{init}}(i,k)\ge\beta q
\]
because each time machine $i$ is whacked, it receives at least $\beta$ units of flow immediately afterward. With an argument analogous to the one used to upper-bound $Z^{\text{new}}(j,k)$, we obtain a lower bound 
\[
Z^{\text{new}}(i,k)\ge(1+\eps)\beta(q-1)-\frac{(1+\eps)2\beta}{\eps}
\]
where the additive term follows from the fact that the machine load $x(i)$ is always at most $2\beta$. This is where we crucially exploited prioritizing whacking over boosting. The two bounds give 
\begin{equation}
Z(i,k)\ge(2+\eps)\beta q-\frac{3\beta}{\eps}.\label{eq:overview-machine-suffix}
\end{equation}
Thus, if $q\ge\alpha=\lceil100/\eps^{2}\rceil$, then 
\begin{equation}
Z(i,k)\ge\left(2+\frac{\eps}{2}\right)\beta q.\label{eq:overview-machine-suffix-long}
\end{equation}
So a long machine suffix receives, on average, noticeably more than $2\beta$ units of history flow per level.

Thus, we established the job-machine flow gap: for a long enough suffix, we have $Z(i,k)\ge\left(2+\frac{\eps}{2}\right)\beta q$ while $Z(j,k)\le\left(2+\frac{\eps}{10}\right)q.$ This creates a constant-factor slack of $1+\Theta(\eps)$ between the two sides.

\paragraph{The final expansion argument on long suffixes in the history graph.}

We now apply the same level-set expansion strategy as in the warm-up, but to long suffixes in the history graph. For a threshold $k$, define 
\[
\widehat{\J}_{\ge k}:=\{(j,r):k\le r\le\ell(j)\text{ and }\ell(j)\ge k+\alpha\},
\]
and define $\widehat{\M}_{\ge k}$ analogously. Thus, $\widehat{\J}_{\ge k}$ contains exactly the job-level nodes that belong to job suffixes of length at least $\alpha+1$ starting from level $k$, and similarly for machines.

The first inequality says that many job-level nodes must feed the high machine suffixes. Consider $\widehat{\M}_{\ge k}$. By~\eqref{eq:overview-machine-suffix-long}, the total incoming history flow of these machine-level nodes is at least 
\[
\left(2+\frac{\eps}{2}\right)\beta\cdot|\widehat{\M}_{\ge k}|.
\]
Ignoring only the globally negligible contribution of tiny history edges, this flow can essentially only come from job-level nodes whose levels are at least $k$. Every such node is contained in $\widehat{\J}_{\ge k-\alpha}$: the extra shift by $\alpha$ ensures that the relevant job suffix is long enough for~\eqref{eq:overview-job-suffix-long}. By~\eqref{eq:overview-job-suffix-long}, the total outgoing history flow of $\widehat{\J}_{\ge k-\alpha}$ is at most 
\[
\left(2+\frac{\eps}{10}\right)|\widehat{\J}_{\ge k-\alpha}|.
\]
Double-counting the same history flow gives 
\[
\left(2+\frac{\eps}{10}\right)|\widehat{\J}_{\ge k-\alpha}|\ge\left(2+\frac{\eps}{2}\right)\beta|\widehat{\M}_{\ge k}|.
\]
Since $\beta=1+10\eps$, the ratio on the right is at least $1+\eps$ times $|\widehat{\M}_{\ge k}|$ after decreasing constants to absorb the negligible tiny-edge term. Hence 
\begin{equation}
|\widehat{\J}_{\ge k-\alpha}|\ge(1+\eps)|\widehat{\M}_{\ge k}|.\label{eq:overview-first-history-expansion}
\end{equation}
This is the history-graph analogue of the first warm-up inequality.

The second inequality uses the current optimum assignment $x^{*}$, just as in the warm-up. Take a node $(j,r)\in\widehat{\J}_{\ge k}$, and let $i$ be the machine to which $j$ is assigned by $x^{*}$. We know that $\ell(i)\ge\ell(j)-L\ge k-L$, otherwise $x(j,i)$ exceeds one. Therefore, for every $r\in[k,\ell(j)]$, the machine-level node $(i,r-L)$ exists and belongs to $\widehat{\M}_{\ge k-L}$. Because $x^{*}$ assigns at most one unit job to each machine, this maps distinct job suffix nodes to distinct machine suffix nodes. Hence 
\begin{equation}
|\widehat{\M}_{\ge k-L}|\ge|\widehat{\J}_{\ge k}|.\label{eq:overview-second-history-expansion}
\end{equation}

Combining~\eqref{eq:overview-first-history-expansion} and~\eqref{eq:overview-second-history-expansion}, with the second inequality applied at threshold $k-L-\alpha$, gives 
\[
|\widehat{\M}_{\ge k-L-\alpha}|\ge|\widehat{\J}_{\ge k-\alpha}|\ge(1+\eps)|\widehat{\M}_{\ge k}|.
\]
Thus, whenever a high machine suffix exists, lowering the threshold by $L+\alpha$ increases the number of counted machine-level nodes by a factor of at least $1+\eps$. This cannot repeat too many times, because the history graph contains only one node per vertex-level pair. Therefore no vertex reaches level larger than $\poly(\log n,1/\eps)$. This proves the key claim and completes the overview of the fractional algorithm.

\subsection{Rounding the Dynamic Fractional Assignment}

\label{sec:overview:round}

The fractional algorithm maintains a $(1+O(\eps))$-approximate fractional assignment. We now explain how to round a dynamic fractional assignment with maximum load $\Gamma$ and maximum job weight $W$ to an integral assignment with maximum load $(1+\eps)\Gamma+W$.

\paragraph{Background: the Shmoys Tardos slot graph.}

The starting point is the classical Shmoys Tardos rounding scheme. Let $G=(J\cup M,E)$ be an input graph with maximum job weight $W$ and $x$ be a fractional assignment with load $\Gamma$. 
Shmoys and Tardos construct a \emph{slot} graph $H=(J\cup S,E_{H})$ whose left side is the set of jobs and whose right side is obtained by replacing each machine $i$ with a set of \emph{slots} $S(i)\subseteq S$. 

The key property is that any job-perfect \emph{matching} in $H$ gives an integral assignment in the original instance: if job $j$ is matched to a slot of machine $i$, then we assign $j$ to $i$. Moreover, the load of the resulting integral assignment is $\Gamma+W$. Thus, using the fractional assignment, they reduce the problem of finding a (weighted) assignment to that of finding an (unweighted) perfect matching. 

\paragraph{The online rounding framework of KLS.}

Krishnaswamy, Li, and Suriyanarayana \cite{Krishnaswamy0S23} showed how to adapt the Shmoys--Tardos rounding scheme to the online setting. Given a graph $G$ and a fractional assignment $x$ undergoing updates, the goal here is to maintain an integral assignment with low \emph{recourse}, i.e., a small number of reassignments. 

They define a similar slot graph $H=(J\cup S,E_{H})$ with a slightly weaker approximation: any job-perfect matching in $H$ yields an integral assignment of load $(1+O(\eps))\Gamma+W$. However, their slot graph $H$ has nicer structural properties:
\begin{enumerate}
\item \textbf{(Vertex expansion):} every set $A$ of jobs has at least $(1+\Omega(\eps))|A|$ neighboring slots, and 
\item \textbf{(Stable slots)}: the number of slot insertions and deletions is bounded in terms of the number of updates to $x$.\footnote{In fact, it is $O(1/\eps^{2})$ times the total $\ell_{1}$ change of $x$.}
\end{enumerate}
The algorithm of \cite{Krishnaswamy0S23} simply maintains the slot graph $H$ and a job-perfect matching $M_{H}$ in $H$. Whenever $M_{H}$ is not perfect, they simply augment through an $O(\frac{\log n}{\eps})$-length augmenting path; such a short path is guaranteed to exist by the vertex expansion of $H$ and induces only $O(\frac{\log n}{\eps})$ recourse. Vertices in $M_{H}$ are not unmatched often because slots are stable. These facts imply a left-perfect matching with small recourse, but not yet with small \emph{update time}.

\paragraph{Barriers to making KLS fast.}

There are two barriers. 
\begin{enumerate}
\item \textbf{(Representation barrier)}: Edges of $H$ are not stable. A small change in the fractional assignment $x$ can change many edges of the slot graph $H$. Thus explicitly maintaining all edges of $H$ will be too expensive.
\item \textbf{(Algorithmic barrier): }Even if we ignore the representation issue, we do not know a fast fully dynamic algorithm for maintaining a left-perfect matching in a changing bipartite graph merely from the promise that the graph has vertex expansion $1+\eps$. Expansion guarantees that a perfect matching exists, but it does not directly tell us how to find short augmenting paths quickly after dynamic updates.
\end{enumerate}

\paragraph{Our strategy: random augmenting path based on fractional matching with gap.}

Based on the fractional assignment $x$ in $G$, \cite{Krishnaswamy0S23} also defined a fractional assignment $y$ with a \emph{gap }guarantee: each job $j$ sends  $y(j)=1$ unit and each slot $s$ receives at most $y(s)<1-2\eps$ units.
They used $y$ in the \emph{analysis} to prove the vertex expansion of $H$ and the approximation guarantee of a left-perfect matching in $H$.

Our key observation is that $y$ is, in fact, extremely helpful \emph{algorithmically}, not just in analysis. Moreover, although $y$ changes substantially under updates, as $H$ does, we can still \emph{implicitly} maintain it.

Suppose for the moment that we have oracle access to $y$. We can explicitly maintain a left-perfect matching $M_{H}$ in the slot graph $H$ undergoing fully dynamic updates with only $\tilde{O}(1/\eps)$ update time. Whenever a job becomes unmatched, we can very quickly repair the matching by finding a \emph{random augmenting path }from this unmatched job: Starting from the current unmatched job $j$, we sample a neighboring slot $s$ with probability proportional to $y(j,s)$. If $s$ is unmatched, we have found an augmenting path and we augment along it. If $s$ is matched, we follow the matching edge from $s$ to the job currently matched to $s$, and then repeat the same sampling step from that job. 

The reason this walk is short is exactly the gap in $y$: each job sends at least $1$ unit, but each slot receives at most $1-2\eps$ units. Intuitively, when the walk samples according to $y$, the matched slots cannot absorb all of the probability forever, because the slot side has slack. At each step, a constant $\Omega(\eps)$ fraction of the remaining probability reaches unmatched slots. Therefore the augmenting-path length is $O(\log(n)/\eps)$ with high probability. This bounds the update time of $\tilde{O}(1/\eps)$ as well. 

Thus, the fractional assignment $y$ with a gap guarantee is stronger than merely the graph $H$ with vertex expansion. This allows us to maintain a left-perfect matching and overcome the algorithmic barrier above.

\paragraph{Dynamic rounding without explicit slot graph.}

The simplicity of the random augmenting path algorithm allows us to work without explicitly storing the full slot graph $H$ or the fractional assignment $y$, overcoming the representation barrier too. Indeed, the random walk needs only very limited access to $y$: given a job $j$, sample a slot $s$ with probability proportional to $y(j,s)$; given a slot $s$, determine whether it is matched and, if so, to which job; and after an update, identify the small set of slots or matched jobs whose status may have become invalid.

These operations can be supported implicitly. For each machine $i$, its slots $S(i)$ form an ordered set of intervals. Using standard interval data structures, given a job-machine pair $(j,i)$, we can find $O(1)$ slots $s$ of machine $i$ where $(j,s)\in E_{H}$, as well as sample a neighbor of $j$ according to the $y$-distribution.

Thus the algorithm never materializes all edges of $H$. The explicit object we maintain is only the integral matching $M_{H}$, plus compact interval data structures that allow the required sampling queries. Updates to the fractional assignment may implicitly change many edges of $H$, but will lead to few slot updates because slots are stable. Each slot update can unmatch one job. Each unmatched job can be repaired very quickly by a short random augmenting path. Hence, the final update time is small.

\paragraph{Putting the pieces together.}

The fractional algorithm supplies a dynamic fractional assignment with load $1+O(\eps)$ and bounded total change. The rounding algorithm treats this fractional assignment as a black box, maintains the implicit slot graph $H$ and the matching $M_{H}$, and repairs the matching $M_{H}$ using short random augmenting paths guided by $y$. Combining the two parts gives a decremental algorithm for load balancing under edge deletions with near-linear total update time and integral load $(1+O(\eps))\Gamma+W$. 

We note that the rounding framework is completely black-box and could be applied to other dynamic fractional load-balancing algorithms as well, whether decremental, incremental, or fully dynamic. For example, we can apply it to the incremental algorithm of Bhattacharya, Kiss, and Saranurak \cite{BhattacharyaKS23}. 

%% file: prelim.tex
\section{Preliminaries}
\label{sec:prelim}

\paragraph{Input graph.}
Consider a bipartite graph $G = (J \cup M, E)$, where the nodes in $J$ (resp.~$M$) are referred to as {\em jobs} (resp.~{\em machines}), and a function $w : J \rightarrow [1, W]$ where $w(j)$ denotes the {\em weight} of a job $j \in J$, for some parameter $W > 1$.  We use the index $j$ (resp.~$i$) to denote a job (resp.~machine). For any subset of jobs $J' \subseteq J$, we define $w(J') := \sum_{j \in J'} w(j)$. For any node $v \in J \cup M$, let $N_G(v) := \{ u \in J \cup M : (u, v) \in E\}$ denote the set of neighbors of $v$ in $G$. Whenever the graph $G$ is clear from the context, we omit the subscript $G$ and simply write $N(v)$ instead of $N_G(v)$. 

\paragraph{Flow.}
We refer to any function $x : E \rightarrow \mathbb{R}_{\geq 0}$ as a {\em flow} in $G$. In particular, we say that $x(j, i)$ is the {\em flow on the edge} $(j, i) \in E$ in $x$; it specifies the fractional extent to which $x$ assigns the job $j$ to the machine $i$. For all jobs $j \in J$, we define $x(j) := \sum_{i \in N(j)} x(j, i)$ to be the {\em total flow sent out by $j$}. Similarly, for all machines $i \in M$, we define $x(i) := \sum_{j \in N(i)} x(j, i)$ to be the {\em total flow absorbed by $i$}. 

\paragraph{Weighted flow.}
For each $(j, i) \in E$, we refer to the quantity $w(j) \cdot x(j, i)$ as the {\em weighted flow} on the edge $(j, i)$. For all jobs $j \in J$, we define $w_x(j) := \sum_{i \in N(j)} w(j) \cdot x(j, i)$ to be the {\em total weighted flow sent out by $j$}. Similarly, for all machines $i \in M$, we define $w_x(i) := \sum_{j \in N(i)} w(j) \cdot x(j, i)$ to be the {\em total weighted flow absorbed by $i$}. We also refer to $w_x(i)$ as the {\em load of a machine $i \in M$} under $x$. The {\em load of $x$} is defined to be $\max_{i \in M} w_x(i)$, i.e., it equals the maximum load of any machine. 

\paragraph{Assignment.}
We say that a flow $x$ is an {\em assignment} iff  $x(j) = 1$ for all jobs $j \in J$.  We say that $x$ is an {\em integral assignment} iff $x$ is an assignment and $x(j, i) \in \{0, 1\}$ for every edge $(j, i) \in E$. We use the phrase {\em fractional assignment} to refer to an assignment that might or might not be integral. 

We say that $x : E \rightarrow \mathbb{R}_{\geq 0}$ is an {\em $\epsilon$-fractional assignment} iff $x(j) \geq (1+\epsilon)^{-1}$ for every job $j \in J$. 

\paragraph{Load balancing instance.}
In the {\em load balancing} problem, we receive $\left(G = (J \cup M, E), w\right)$ as input, and our goal is to compute an {\em integral assignment} with minimum load. We let $\opt_f(G, w)$ and $\opt(G, w)$ respectively denote the load of an optimal fractional assignment in $(G, w)$ and the load of an optimal integral assignment in $(G, w)$.

\paragraph{Dynamic setting.}
At preprocessing, we receive an input $(G, w)$, and subsequently this input goes through a sequence of {\em restricting updates}: Each update either deletes an existing edge $(j, i) \in E$, or inserts a job $j$ and specifies its weight $w(j)$ and all its incident edges. 
Throughout the rest of this paper, we let $m$ and $n$ respectively denote the number of edges that ever appear and the total number of jobs and machines present in the input at the end of the last update. For ease of presentation, we will assume that our algorithm knows the value of $n$ in advance.\footnote{We can easily modify our algorithm to work without this assumption by guessing the value of $n$. Whenever the actual value of $n$ exceeds this guess, we double our guess and restart the algorithm from scratch on the current input. This incurs only an extra logarithmic-factor overhead in our total update time.}

\paragraph{Parameters.} For the rest of the paper, we define the following  parameters: 
\begin{equation}
\label{eq:epsilon}
\epsilon \in (0, 1/100), \ \ 
\alpha := \frac{5}{\epsilon^3},  \ \ \beta := 1+10\epsilon, \ \ L :=  \log_{(1+\epsilon)} \left(\frac{W m n}{\epsilon}\right)^2, \ \ \kappa :=  \log_{(1+\epsilon/2)} \left( \frac{Wmn}{\epsilon} \right)^6.
\end{equation} For ease of exposition, throughout the rest of the paper we will assume that $L$ and $\kappa$ are integers. 

\thatchaphol{Note for the next version: We can set $\alpha = 10^{10}/\epsilon^2$ and $  L :=  \log_{(1+\epsilon)} \left(\frac{W m n}{\epsilon^2}\right)$ and say that the maximum level is $L(L+\alpha)$. We should be able to remove the dependency on $W$ as well.}


\section{Core Components}

Our dynamic algorithm will consist of two steps. First, we show how to maintain a good $\epsilon$-fractional assignment, as summarized in \Cref{th:fractional} below. Next, we design an efficient {\em rounding scheme} that converts this $\epsilon$-fractional assignment into an integral assignment, with only a small overhead in the approximation guarantee and total update time.

\begin{theorem}
\label{th:fractional} Consider an input $(G = (J \cup M, E), w)$ to the load balancing problem undergoing a sequence of restricting updates,
and a parameter $\Gamma > 0$. Then, there is a deterministic algorithm that
\begin{itemize}
\item Either maintains an $\epsilon$-fractional assignment $x$ in $(G, w)$ with load at most $\beta (1+\epsilon) \cdot \Gamma$,
\item Or certifies that $\opt_f(G, w) > \Gamma$.
\end{itemize}
The algorithm has a total update time of $\tilde{O}(m/\epsilon^4)$, where the  $\tilde{O}(.)$ notation hides polylogarithmic factors  in $m, n, W$ and $1/\epsilon$. 
\end{theorem}

The guarantees of our rounding scheme are summarized by the following theorem. Note that the rounding scheme functions under fully dynamic vertex and edge updates. We treat the insertion or deletion of a vertex $v$ as $\deg_G(v)$ edge updates to the underlying fractional assignment.

\begin{restatable}{theorem}{thmrounding}
\label{thm:rounding}

Let $(G_t = (J_t \cup M_t, E_t), w_t)$ be a dynamic load balancing instance with maximum job weight $W$. Let $x_t$ denote an $\epsilon$-fractional assignment at time $t$ with maximum load $\Gamma_t$. Over a sequence of updates to the instance (which could consist of both edge and vertex insertions and deletions), let $U$ denote the total number of edge updates to the fractional assignment $x_t$.

There exists a randomized dynamic algorithm that, given access to $G_t, w_t, x_t$, maintains an integral assignment of $(G_t,w_t)$ with maximum load at most $(1 + O(\epsilon)) \Gamma_t + W$ at time $t$. The algorithm has total update time $\tilde{O}((m + U)/\epsilon^3)$ with high probability and works against an adaptive adversary.
\end{restatable}

\paragraph{Deriving the main theorems.}
Run the algorithms of \Cref{th:fractional,thm:rounding} with accuracy
$\delta=c\epsilon$, where $c>0$ is a sufficiently small constant. In
parallel, run \Cref{th:fractional} for the guesses
\[
    \Gamma_r:=\frac{(1+\delta)^r}{n},
    \qquad
    0\le r\le \left\lceil\log_{1+\delta}(n^2W)\right\rceil.
\]
For every feasible nonempty instance, $1/n\le\opt_f\le nW$. Hence, if
$\widehat\Gamma$ is the smallest guess not certified as too small by
\Cref{th:fractional}, then
\begin{equation}
    \label{eq:smallest-valid-guess}
    \widehat\Gamma\le(1+\delta)\opt_f.
\end{equation}
There are $O(\delta^{-1}\log(nW))$ guesses, so running all instances only
increases the total update time by a factor polynomial in
$\log(nW)/\epsilon$.

\begin{proof}[Proof of \Cref{thm:intro-fractional}]
Let $x$ be the $\delta$-fractional assignment for
$\widehat\Gamma$, and normalize it job by job:
$\bar x(j,i):=x(j,i)/x(j)$. Then $\bar x$ is a fractional assignment and
its load is at most $(1+\delta)$ times that of $x$. Since
$\beta=1+10\delta$, \Cref{th:fractional} and
\eqref{eq:smallest-valid-guess} imply that its load is at most
\[
    \beta(1+\delta)^3\opt_f=(1+O(\delta))\opt_f
    \le(1+\epsilon)\opt_f,
\]
where the last inequality follows by choosing $c$ sufficiently small. The
claimed total update time follows from the bound on the number of guesses.
\end{proof}

\begin{proof}[Proof of \Cref{thm:intro-integral}]
For every guess, apply \Cref{thm:rounding} to the fractional assignment
maintained by \Cref{th:fractional}. The number of edge-value changes in
each such assignment is at most its total update time, so
$U=\widetilde O(m/\delta^4)$. Thus, \Cref{thm:rounding}, summed over all
guesses, still takes total time
$m\poly(\log(nW)/\epsilon)$ with high probability. Output the rounded
assignment for $\widehat\Gamma$. Its load is at most
\[
    (1+O(\delta))\beta(1+\delta)\widehat\Gamma+W
    \le (1+\epsilon)\opt+W
    \le (2+\epsilon)\opt,
\]
where we used \eqref{eq:smallest-valid-guess}, $\opt_f\le\opt$, and
$W\le\opt$. A union bound over all guesses gives the claimed
high-probability guarantee, including against an adaptive adversary.
\end{proof}

%% file: fractional-algorithm.tex
\section{Maintaining an $\epsilon$-Fractional Assignment: The Algorithm}
\label{sec:algo}

Note that if $\Gamma < 1/n$, then we can easily certify that $\opt_f(G, w) > \Gamma$ at preprocessing. Accordingly, throughout the rest of the paper, we assume that
\begin{equation}
\label{eq:gamma:100}
\Gamma \geq 1/n.
\end{equation}

\thatchaphol{TODO: we should remove dependency on $W$ from the running time.}

We will show how to maintain an $\epsilon$-fractional assignment $x$ in an input $(G = (J \cup M, E), w)$ undergoing decremental updates. Our dynamic algorithm repeatedly invokes two simple subroutines. The subroutine $\boost(j, G, x)$  {\em increases} the flow on each edge incident on the job $j \in J$ by a multiplicative $(1+\epsilon)$ factor (see \Cref{alg:boost}). In contrast, the subroutine $\whack(i, G, x)$ {\em decreases} the flow on each edge incident on the machine $i \in M$ by a multiplicative $(1+\epsilon)$ factor (see \Cref{alg:whack}).  Whenever we call the subroutine $\boost(j, G, x)$, we say that the {\bf job $j$ gets boosted}. Similarly, whenever we call the subroutine $\whack(i, G, x)$, we say that the {\bf machine $i$ gets whacked}.

\begin{algorithm}[h]
\caption{$\boost(j, G = (J \cup M, E), x)$, for some job $j \in J$. \label{alg:boost}}
  \DontPrintSemicolon
   \For{every $i \in N(j)$}{
        $x(j, i) \leftarrow (1+\epsilon) \cdot x(j, i)$.
   }
    \Return $x$.
\end{algorithm}

\begin{algorithm}[h]
\caption{$\whack(i, G = (J \cup M, E), x)$, for some machine $i \in M$. \label{alg:whack}}
  \DontPrintSemicolon
   \For{every $j \in N(i)$}{
        $x(j, i) \leftarrow (1+\epsilon)^{-1} \cdot x(j, i)$.
   }
    \Return $x$.
\end{algorithm}

\begin{algorithm}[h]
\caption{$\preprocess(G = (J \cup M, E), w)$ \label{alg:preprocess}}
  \DontPrintSemicolon
  \For{every job $j \in J$}{
        Initialize $\ell(j) \leftarrow 0$.
  }
  \For{every machine $i \in M$}{
        Initialize $\ell(i) \leftarrow 0$.
  }
   \For{every edge $(j, i) \in E$}{
       Initialize $x(j, i) \leftarrow \left(\frac{\epsilon}{Wmn}\right)^2 \cdot (1+\epsilon)^{\ell(j) - \ell(i)} = \left(\frac{\epsilon}{Wmn}\right)^2$. 
        \label{line:level:init:1}
   }
   \While{$x$ is {\em not} an $\epsilon$-fractional assignment in $(G, w)$ with load at most $\beta(1+\epsilon) \cdot \Gamma$}{
   \If{there exists some machine $i \in M$ with $w_x(i) > \beta(1+\epsilon)  \cdot \Gamma$}{ 
            $\whack(i, G, x)$.     \qquad // See \Cref{alg:whack}  \label{line:whack:preprocess}

            $\ell(i) \leftarrow \ell(i) + 1$.

            \If{$\ell(i) \geq 10 \kappa(7L+\alpha)$}{
                {\bf terminate}. \qquad // We will show that this implies $\opt_f(G, w) > \Gamma$.
            }
   }
   \Else{
            There must necessarily exist some job $j \in J$ with $x(j) < (1+\epsilon)^{-1}$. 
            
            $\boost(j, G, x)$.  \qquad // See \Cref{alg:boost}

            $\ell(j) \leftarrow \ell(j) + 1$.
            
            \If{$\ell(j) \geq 10 \kappa(7L+\alpha)$}{
                {\bf terminate}. \qquad // We will show that this implies $\opt_f(G, w) > \Gamma$.
            }
   }
   }
    
    \Return $x$.
\end{algorithm}

We define the {\bf level of a job} (resp.~{\bf machine}) to be the number of times it has been boosted (resp.~whacked) until now. Let $\ell(j)$ and $\ell(i)$ respectively denote the levels of a job $j \in J$ and a machine $i \in M$. At the start of preprocessing, we initialize $\ell(j) = \ell(i) = 0$ for every job $j \in J$ and machine $i \in M$. While dealing with the sequence of updates, whenever a job $j$ gets inserted, we initialize $\ell(j) = 0$. Subsequently, the level of a job (resp.~machine) increases monotonically with time as it keeps getting boosted (resp.~whacked). 

\begin{algorithm}[h]
\caption{$\delete$ \label{alg:delete-edge}}
  \DontPrintSemicolon
  \If{the update consists of the deletion of an edge $(j, i)$ from $E$}{   
        Delete the edge $(j, i)$ from the input.
   }
   \ElseIf{the update consists of the insertion of a job $j$}{
        Initialize $\ell(j) \leftarrow 0$.
        
        \For{all $i \in N(j)$}{
                $x(j, i) \leftarrow \left(\frac{\epsilon}{Wmn}\right)^2 \cdot (1+\epsilon)^{\ell(j) -\ell(i)} = \left(\frac{\epsilon}{Wmn}\right)^2 \cdot (1+\epsilon)^{-\ell(i)}$. \label{line:level:init:2}
            }
   }
   \While{$x$ is {\em not} an $\epsilon$-fractional assignment in $(G, w)$ with load at most $\beta(1+\epsilon) \cdot \Gamma$}{
   \If{there exists some machine $i \in M$ with $w_x(i) > \beta(1+\epsilon)  \cdot \Gamma$}{ 
            $\whack(i, G, x)$.     \qquad // See \Cref{alg:whack}   \label{line:whack:update}

            $\ell(i) \leftarrow \ell(i) + 1$.

            \If{$\ell(i) \geq 10\kappa(7L+\alpha)$}{
                {\bf terminate}. \qquad // We will show that this implies $\opt_f(G, w) > \Gamma$.
            }
   }
   \Else{
            There must necessarily exist some job $j \in J$ with $x(j) < (1+\epsilon)^{-1}$. 
            
            $\boost(j, G, x)$.  \qquad // See \Cref{alg:boost}

            $\ell(j) \leftarrow \ell(j) + 1$.
            
            \If{$\ell(j) \geq 10 \kappa(7L+\alpha)$}{
                {\bf terminate}. \qquad // We will show that this implies $\opt_f(G, w) > \Gamma$.
            }
   }
   }
\end{algorithm}

At preprocessing, we receive the input $(G = (J \cup M, E), w)$, and we initialize the flow $x$ by calling \Cref{alg:preprocess}. 
The algorithm starts by setting a flow of $\epsilon^2/(Wmn)^2$ on each edge. 
Next, as long as $x$ is not an $\epsilon$-fractional assignment with load at most $\beta(1+\epsilon) \cdot \Gamma$, we repeatedly perform the following operation: If there is a machine $i \in M$ with load more than $\beta(1+\epsilon) \cdot \Gamma$, then we whack the machine $i$. Otherwise, there must exist  some job $j \in J$ with $x(j) < (1+\epsilon)^{-1}$,  and we boost the job $j$.  Note that we always prioritize whacking machines over boosting jobs, and this tie-breaking rule will play a crucial role in our analysis. Furthermore, 
whenever the level of any job or machine becomes $\geq 10 \kappa(7L+\alpha)$, 
we terminate the algorithm. We will show that this occurs only if $\opt_f(G, w) > \Gamma$.

We handle each update as described in \Cref{alg:delete-edge}. In a bit more detail, if the update involves insertion of a new job $j$, 
then we initialize $\ell(j) = 0$ and a flow $x(j, i) = \left(\frac{\epsilon}{Wmn}\right)^2 \cdot (1+\epsilon)^{-\ell(i)}$ on each edge $(j, i) \in E$ incident on $j$. 
In contrast, if the update involves the deletion of an edge $(j, i)$, then we simply remove the flow $x(j, i)$. 
After this, we perform the same operations as we did at preprocessing, until either  $x$ is an $\epsilon$-fractional assignment with load at most $\beta(1+\epsilon) \cdot \Gamma$, or 
the level of any job or machine becomes $\geq 10 \kappa(7L+\alpha)$. 
In the latter case, we terminate the algorithm and declare that $\opt_f(G, w) > \Gamma$.

\medskip

We conclude this section with a few crucial observations.

\begin{observation}
\label{obs:flow}
The flow on an edge $(j, i) \in E$ is always equal to $x(j, i) = \left(\frac{\epsilon}{Wmn}\right)^2 \cdot (1+\epsilon)^{\ell(j) - \ell(i)}$. 
\end{observation}

\begin{proof}
The observation clearly holds when we first initialize the flow on this edge (see Line~\ref{line:level:init:1} in \Cref{alg:preprocess} and Line~\ref{line:level:init:2} in \Cref{alg:delete-edge}). Subsequently, whenever we boost the job $j$, the flow $x(j, i)$ increases by a multiplicative $(1+\epsilon)$ factor and the level $\ell(j)$ increases by an additive one. In contrast, whenever we whack the machine $i$, the flow $x(j, i)$ decreases by a multiplicative $(1+\epsilon)$ factor and the level $\ell(i)$ increases by an additive one. Thus, the observation continues to hold at all times.
\end{proof}

\begin{observation}
\label{obs:jobs:flow:upper:bound}
For every job $j \in J$, we always have $0 \leq x(j) < 1$.
\end{observation}

\begin{proof} 
If the job $j$ is present at preprocessing, 
then we initialize $x(j) = \sum_{i \in N(j)} x(j,i) = |N(j)| \cdot \left(\frac{\epsilon}{Wmn}\right)^2 \leq m \cdot \left(\frac{\epsilon}{Wmn}\right)^2 \leq 1$ (see Line~\ref{line:level:init:1} in \Cref{alg:preprocess}). 
Otherwise, when it is inserted, 
we initialize $x(j) = \sum_{i \in N(j)} x(j,i) = \sum_{i \in N(j)} \left(\frac{\epsilon}{Wmn}\right)^2 \cdot (1+\epsilon)^{-\ell(j)} \leq  |N(j)| \cdot \left(\frac{\epsilon}{Wmn}\right)^2  \leq 1$ (see Line~\ref{line:level:init:2} in \Cref{alg:delete-edge}). 

Subsequently, the value of $x(j)$ increases only when the job $j$ gets boosted. We decide to boost the job $j$ only if $x(j) < (1+\epsilon)^{-1}$. On the other hand, whenever we boost the job $j$, the value of $x(j)$ increases by a multiplicative $(1+\epsilon)$ factor. Thus, we always continue to have $x(j) < 1$. 
\end{proof}

\begin{corollary}
\label{obs:flow:upper:bound}
For every edge $(j, i) \in E$, we always have $0 \leq x(j, i) < 1$. 
\end{corollary}

\begin{proof}
Follows from \Cref{obs:jobs:flow:upper:bound}, since $x(j) = \sum_{i \in N(j)} x(j, i)$.
\end{proof}

\begin{observation}
\label{assume:label}
For all $i \in M$ and $j \in J$, 
we have $\ell(i) \leq 10 \kappa(7L+\alpha)$ and $\ell(j) \leq 10\kappa(7L+\alpha)$. 
\end{observation}

\begin{proof}
We terminate 
when a job or a machine reaches level $10 \kappa(7L+\alpha)$ (see \Cref{alg:preprocess} and  \Cref{alg:delete-edge}).
\end{proof}

\section{Maintaining an $\epsilon$-Fractional Assignment: The Analysis}
\label{sec:analysis}

We analyze the update time of our algorithm in the lemma below.

\begin{lemma}
\label{lm:main:time}
Our dynamic algorithm, as described in \Cref{sec:algo}, has a total update time of $\tilde{O}(m/\epsilon^4)$. 
\end{lemma}

\begin{proof}
Using standard data structures, our total update time is dominated (up to a polylogarithmic factor) by the time spent on the calls to the $
\boost(j, G, x)$ and $\whack(i,G, x)$ subroutines. Furthermore, a call to $\boost(j, G, x)$ takes time proportional to the current degree of the job $j$ in $G$, whereas a call to $\whack(i, G, x)$ takes time proportional to the current degree of the machine $i$ in $G$. \peter{As the pseudocode is written, it is not clear how the jobs/machines to be boosted/whacked are found. It is trivial, as they must be a neighbor of a boosted/whacked node, but maybe worth mentioning.}

Suppose that whenever we make a call to $\boost(j, G, x)$ (resp.~$\whack(i,G,x)$), we place one unit of token on each of the edges currently incident on the job $j$ (resp.~the machine $i$). Evidently, 
we terminate the algorithm whenever the level of any job or machine becomes equal to $10 \kappa(7L+\alpha)$ (see \Cref{assume:label}). 
Hence, throughout the entire duration of our algorithm, each edge that is ever present in the input receives 
at most $O(\kappa(L+\alpha)) = \tilde{O}(1/\epsilon^4)$ tokens (see \Cref{eq:epsilon}). This implies that the total number of tokens that ever get created is at most $\tilde{O}(m/\epsilon^4)$. 
As discussed above, these tokens are sufficient to pay  for the total update time of our algorithm. This concludes the proof.
\end{proof}

Our key technical contribution is to prove the following lemma.

\begin{lemma}
\label{lm:upper:bound:level} Throughout the sequence of updates, as long as $\opt_f(G) \leq \Gamma$, 
it must necessarily be the case that $\ell(j) < 10 \kappa(7L+\alpha)$ for all jobs $j \in J$ and $\ell(i) < 10 \kappa(7L+\alpha)$ for all machines $i \in M$. 
\end{lemma}

It is trivial to verify that as long as the levels of all jobs and machines 
remain below $10\kappa(7L+\alpha)$ after handling the current update, 
our algorithm  maintains an $\epsilon$-fractional assignment $x$ in $(G, w)$ with load at most $\beta(1+\epsilon) \cdot \Gamma$. Accordingly, \Cref{th:fractional} follows from \Cref{lm:main:time} and \Cref{lm:upper:bound:level}.

\paragraph{Roadmap for the rest of this section.} We devote the rest of \Cref{sec:analysis} towards proving \Cref{lm:upper:bound:level}. We start by introducing the concepts of an {\em history graph} $\G$ and its {\em history flow} $z$ in \Cref{sec:history-graph}. We emphasize that these concepts are used only for the sake of analysis. Our algorithm, as described in \Cref{sec:algo}, is oblivious to them. Finally, we derive the proof of \Cref{lm:upper:bound:level} in \Cref{sec:history-graph:properties}.

\subsection{The History Graph $\G = (\J \cup \M, \E)$ and the History Flow $z : \E \rightarrow \mathbb{R}_{\geq 0}$}
\label{sec:history-graph}

\paragraph{Time-steps.} We think of time as consisting of a sequence of discrete steps $t \in \mathbb{Z}_{\geq 0}$.\footnote{Throughout this paper, the symbol $\mathbb{Z}_{\geq 0}$ denotes the set of nonnegative integers.} Time-step $t = 0$ refers to the scenario before preprocessing. Each subsequent time-step $t > 0$ corresponds to one of the following types of events: (i) deletion of an edge, (ii) insertion of a job, (iii) boosting a job, and (iv) whacking a machine. We use the superscript $t$ to denote the status of any object at the {\em end} of time-step $t$. 
We are now ready to define the history graph and its history flow.

\begin{definition}
\label{def:history-graph}
At the end of  time-step $t \in \mathbb{Z}_{\geq 0}$, the {\em history graph} $\G = (\J \cup \M, \E)$ is given by:
\begin{eqnarray*}
\J^{(t)} & := & \left\{ (j, k_j) \in J^{(t)} \times \mathbb{Z}_{\geq 0} : k_j \leq \ell^{(t)}(j) \right\}.  \\
\M^{(t)} & := & \left\{ (i, k_i) \in M^{(t)} \times \mathbb{Z}_{\geq 0} : k_i \leq \ell^{(t)}(i) \right\}. \\
\E^{(t)} & := & \left\{ ((j,k_j), (i, k_i)) \in \J^{(t)} \times \M^{(t)} : \exists \, t' \in [0, t] \text{ s.t. } (j, i) \in E^{(t')}, k_j = \ell^{(t')}(j), k_i = \ell^{(t')}(i) \right\}. 
\end{eqnarray*}
Furthermore, 
the {\em history flow} $z: \E \rightarrow \mathbb{R}_{\geq 0}$ at the end of time-step $t \in \mathbb{Z}_{\geq 0}$ is given by:
$$z^{(t)}\big((j,k_j), (i, k_i)\big)  =  \left(\frac{\epsilon}{Wmn}\right)^2 \cdot (1+\epsilon)^{k_j - k_i} \text{ for every edge } ((j,k_j), (i, k_i)) \in \E^{(t)}.$$ 
We refer to each node $(j, k_j) \in \J$ as an {\bf job-level node at level $k_j$},  each node $(i, k_i) \in \M$ as an {\bf machine-level node at level $k_i$}, and each edge $\big((j,k_j), (i, k_i)\big) \in \E$ as an {\bf history edge}. We also abuse notation and define $w : \J \rightarrow [1, W]$ such that $w(j, k_j) = w(j)$ for all $(j, k_j) \in \J$.
\end{definition}

\begin{observation}
\label{obs:history-edge}
Consider any time-step $t' \in \mathbb{Z}_{\geq 0}$, and any edge $e = (j, i) \in E^{(t')}$ with $\ell^{(t')}(j) = k_j$ and $\ell^{(t')}(i) = k_i$. Then, at every future time-step $t \geq t'$, there exists a history edge $e_{\mathrm{hist}} = \big( (j, k_j), (i, k_i) \big) \in \E^{(t)}$ with history flow $z^{(t)}(e_{\mathrm{hist}}) = x^{(t')}(e)$. Furthermore, these are the only history edges that are present in $\G$ at any given time-step. \peter{This sentence is correct, but on a first glance suggests that the state of $G$ at a given time step describes $\mathcal{G}$, which is not the case}
\end{observation}

\begin{proof}
Follows from \Cref{def:history-graph} and \Cref{obs:flow}.
\end{proof}

By \Cref{obs:history-edge}, the history graph $\G$ and the history flow $z$ evolve {\em monotonically} over time, in contrast with the input $(G, w)$ and the flow $x$. This is summarized in the corollary below.

\begin{corollary}
\label{cor:monotone}
No job-level node, machine-level node or history edge ever gets deleted from $\G$. Moreover, the history flow $z(e_{\mathrm{hist}})$ on a history edge $e_{\mathrm{hist}}$ gets fixed at the time  $e_{\mathrm{hist}}$ is inserted into $\G$, and it never changes subsequently in the future. 
\end{corollary}

\paragraph{Notational Conventions.} While dealing with  the history graph $(\G, w)$ and the history flow $z$, we will use the same notational conventions introduced in \Cref{sec:prelim}. For example, the symbol $$w_z(i, k_i) := \sum_{(j, k_j) \in N_{\G}(i, k_i)} w(j, k_j) \cdot z\big( (j, k_j), (i, k_i) \big) = \sum_{(j, k_j) \in N_{\G}(i, k_i)} w(j) \cdot z\big( (j, k_j), (i, k_i) \big)$$ denotes the load of a machine-level node $(i, k_i)$. Whenever we use the phrases ``{\em total weighted flow absorbed by the machine-level node $(i, k_i)$'}' and ``{\em total flow sent out by the job-level node $(j, k_j)$}'', we respectively refer to  $w_z(i, k_i)$ and $z(j, k_j)$. In contrast, the phrases ``{\em total weighted flow absorbed by the machine $i$}'' and ``{\em total flow sent out by the job $j$}'' respectively refer to $w_x(i)$ and $x(j)$.

\subsection{Establishing the Correctness of Our Algorithm: Proof of \Cref{lm:upper:bound:level}}
\label{sec:history-graph:properties}

\paragraph{Key Notations.} We introduce some notations that will be used throughout the rest of this section.  
For every level $k \in \mathbb{Z}_{\geq 0}$,  we define  
\begin{eqnarray*}
J_{\geq k} & := & \left\{j \in J : \ell(j) \geq k\right\}. \\  
\J_{\geq k} & := & \left\{ (j, k_j) \in \J : k_j \geq k \right\}. \\
M_{\geq k} & := & \left\{ i \in M  : \ell(i) \geq k \right\}. 
\\ \M_{\geq k} & := &  \left\{ (i, k_i) \in \M : k_i \geq k \right\}.
\end{eqnarray*}
Next, recall the parameter $\alpha$ from \Cref{eq:epsilon}. For every $k \in \mathbb{Z}_{\geq 0}$, we define
\begin{eqnarray*}
\hat{J}_{\geq k} & := & \left\{j \in J_{\geq k} : \ell(j) \geq k+\alpha-1\right\}. \\ \hat{\J}_{\geq k} & := & \left\{ (j, k_j) \in \J_{\geq k} : j \in \hat{J}_{\geq k} \right\}. \\ 
\hat{M}_{\geq k} & := & \left\{ i \in M_{\geq k} : \ell(i) \geq k+\alpha-1 \right\}. \\  \hat{\M}_{\geq k} & := & \left\{ (i, k_i) \in \M_{\geq k} : i \in \hat{M}_{\geq k} \right\}.
\end{eqnarray*}

In words, the set $\hat{J}_{\geq k}$ consists of all jobs in $J_{\geq k}$ whose levels are at least $\alpha-1$ larger than the threshold $k$, and the set $\hat{\J}_{\geq k}$ consists of all the job-level nodes in $\J_{\geq k}$ that originate from  $\hat{J}_{\geq k}$. Thus, each job in $\hat{J}_{\geq k}$ contributes at least $\alpha$ job-level nodes to the set $\hat{\J}_{\geq k}$. The symbols $\hat{M}_{\geq k}$ and $\hat{\M}_{\geq k}$ are defined analogously.  It is easy to verify that
\begin{equation}
\label{eq:equivalence}
\left| \hat{\J}_{\geq k} \right| = \sum_{j \in \hat{J}_{\geq k}} \left(\ell(j) - k + 1 \right) \text{ and }
\left| \hat{\M}_{\geq k} \right| = \sum_{i \in \hat{M}_{\geq k}} \left(\ell(i) - k + 1 \right).
\end{equation}

We next classify each history edge as either {\bf big} or {\bf small}, as follows.
\begin{definition}
\label{def:big}
A history edge $e_{\mathrm{hist}} \in \E$ is {\bf big} if 
$z(e_{\mathrm{hist}}) \geq \left(\frac{\epsilon}{Wmn}\right)^{14}$, 
and {\bf small} otherwise.
\end{definition}

\paragraph{Roadmap for the Rest of \Cref{sec:history-graph:properties}.} 
We start by showing that the levels of any two neighboring jobs and machines are very close to each other, for most edges. Specifically, \Cref{cl:edge:1} implies that $\ell(i) \geq \ell(j) - L$ for all $(j, i) \in E$. In contrast, \Cref{cl:edge:2} implies that as long as any given edge $(j, i) \in E$ carries nontrivial amount of flow, we must have $\ell(j) \geq \ell(i) - 6L$. 

Next, in \Cref{cor:new:lm:job} we derive that for every level $k_j > 0$, each job-level node $(j, k)  \in \hat{\J}_{\geq k_j}$ sends out at most $2(1+\epsilon^2/10) \cdot w(j)$ units of weighted flow, on average. \Cref{cor:new:lm:job} is obtained from \Cref{new:lm:job}. The proof of \Cref{new:lm:job} is delegated to \Cref{sec:new:lm:job}.

In a similar vein, \Cref{cor:edge:2} shows that for every level $k_i > 6L + \alpha$, each machine-level node $(i, k) \in \hat{\M}_{\geq k_i}$ receives at least $2(1+\epsilon - \epsilon^2) \cdot \Gamma$ units of weighted flow, on average, from the job-level nodes in $\hat{\J}_{\geq k_i - 6L - \alpha}$. \Cref{cor:edge:2}, whose proof crucially uses \Cref{cl:edge:2}, is obtained from \Cref{new:lm:machine}. We delegate the proof of \Cref{new:lm:machine} to \Cref{sec:new:lm:machine}. 

 \Cref{new:lm:expansion:100} presents the first part of our {\em expansion argument}. In a bit more detail, \Cref{new:lm:expansion:100} upper bounds the number of machine-level nodes in $\hat{\M}_{\geq k}$ in terms of the total weight of the job-level nodes in $\hat{\J}_{\geq k-6L-\alpha}$. The proof of \Cref{new:lm:expansion:100} relies upon \Cref{cor:new:lm:job} and \Cref{cor:edge:2}.

 \Cref{new:cl:expansion} presents the second part of our expansion argument, by complementing  \Cref{new:lm:expansion:100}. Intuitively, \Cref{new:cl:expansion} upper bounds the total weight of the job-level nodes in $\hat{\J}_{\geq k}$ in terms of the number of machine-level nodes in $\hat{\M}_{\geq k-L}$. 

 In \Cref{cor:expansion}, we combine the statements of \Cref{new:lm:expansion:100} and \Cref{new:cl:expansion}, to argue that the size of the set $\hat{\M}_{\geq k}$ expands geometrically 
 whenever we decrease the value of $k$ by an additive $7L+\alpha$. 

 At the end of \Cref{cor:expansion}, we finally present the proof of \Cref{lm:upper:bound:level}. This proof relies upon the insight that the geometric expansion guaranteed by \Cref{cor:expansion} above cannot be repeated for more than logarithmic many iterations. 


\begin{claim}
\label{cl:edge:1}
For every history edge $\big( (j, k_j), (i, k_i)\big) \in \E$, we have $k_i \geq k_j - L$.
\end{claim}

\begin{proof}
Consider any time-step $t \in \mathbb{Z}_{\geq 0}$ and any history edge $e_{\mathrm{hist}} = \big( (j, k_j), (i, k_i)\big) \in \E^{(t)}$. Let $e = (j, i)$. By \Cref{obs:history-edge}, there is a time-step $t' \in [0, t]$ where $z^{(t)}(e_{\mathrm{hist}}) = x^{(t')}(e)$. By \Cref{obs:flow:upper:bound}, we have $x^{(t')}(e) < 1$. 
Finally, by \Cref{def:history-graph}, we have $z^{(t)}(e_{\mathrm{hist}}) = \left(\epsilon/(Wmn)\right)^2 \cdot (1+\epsilon)^{k_j - k_i}$. Putting everything together, we get $\left(\epsilon/(Wmn)\right)^2 \cdot (1+\epsilon)^{k_j - k_i} < 1$, which implies that  $k_j - k_i \leq \log_{(1+\epsilon)} (Wmn/\epsilon)^2 = L$ (see \Cref{eq:epsilon}).
\end{proof}

\begin{claim}
\label{cl:edge:2}
For every big history edge $\big( (j, k_j), (i, k_i)\big) \in \E$, we have  
$k_j \geq k_i - 6L$. 
\end{claim}

\begin{proof}
Consider any big history edge $e_{\mathrm{hist}} = \big( (j, k_j), (i, k_i)\big) \in \E$. Since $e_{\mathrm{hist}}$ is big, it follows that $z(e_{\mathrm{hist}}) \geq \left(\epsilon/(Wmn)\right)^{14}$. Moreover,  by \Cref{def:history-graph}, we have $z(e_{\mathrm{hist}}) = \left(\epsilon/(Wmn)\right)^2 \cdot (1+\epsilon)^{k_j - k_i}$. Thus, we conclude that $\left(\epsilon/(Wmn)\right)^{14} \leq \left(\epsilon/(Wmn)\right)^2 \cdot (1+\epsilon)^{k_j - k_i}$, and hence $(1+\epsilon)^{k_i - k_j} \leq \left(Wmn/\epsilon\right)^{12}$. In other words, this means that $k_i - k_j \leq  \log_{(1+\epsilon)}(Wmn/\epsilon)^{12} = 6 L$ (see \Cref{eq:epsilon}).
\end{proof}

\begin{lemma}
\label{new:lm:job}
For every  job $j \in J$ and level $0 < k_j < \ell(j)$, we have $$\sum_{k = k_j}^{\ell(j)} z(j, k) \leq 2 \cdot \left(\ell(j) - k_j+1\right) + \left(\frac{1}{\epsilon}\right).$$
\end{lemma}

We will prove the above lemma in \Cref{sec:new:lm:job}.
From this, we can conclude the following useful observation. Consider a level $k$ and a job $j$ whose current level is significantly higher than $k$, i.e., at least $k+\alpha-1$. 
When we average over the job-level nodes $(j, k),\dots,(j, \ell(j))$, each of them sends at most $2 \left(1 + \epsilon^2/10\right)$ times its weight units of weighted flow. This is summarized in the corollary below.

\thatchaphol{I would write $k$ instead of $k_j$ ($k_j$ is not determined by $j$ and so it is confusing.)}

\begin{corollary}
\label{cor:new:lm:job}
For every level $k_j > 0$ and job $j \in \hat{J}_{\geq k_j}$, we have $$\frac{\sum_{k = k_j}^{\ell(j)} w_z(j, k)}{\left(\ell(j) - k_j + 1\right)} \leq 2\left(1+\frac{\epsilon^2}{10}\right) \cdot w(j).$$ 
\end{corollary}

\begin{proof}
From \Cref{new:lm:job}, we infer that:
\begin{eqnarray*}
\frac{\sum_{k = k_j}^{\ell(j)} w_z(j, k)}{\left(\ell(j) - k_j + 1\right)} & = & w(j) \cdot \frac{\sum_{k = k_j}^{\ell(j)} z(j, k)}{\left(\ell(j) - k_j + 1\right)} \\
& \leq & w(j) \cdot \left( 2 + \frac{1}{\epsilon \left( \ell(j) - k_j+1\right)} \right) \\
& \leq & w(j) \cdot \left( 2 + \frac{1}{\epsilon \alpha} \right), \\
& \leq & w(j) \cdot 2\left(1 + \frac{\epsilon^2}{10} \right),
\end{eqnarray*}
where the penultimate inequality holds since $j \in \hat{J}_{\geq k_j}$ (i.e., $\ell(j) - k_j+1 \geq \alpha$) and the last inequality holds because of \Cref{eq:epsilon}. This concludes the proof.
\end{proof}

\begin{lemma}
\label{new:lm:machine}
For every machine $i \in M$ and level $0 < k_i < \ell(i)$, we have 
$$\sum_{k = k_i}^{\ell(i)} w_z(i, k) \geq 2 (1+\epsilon) \cdot \left(\ell(i) - k_i+1\right) \cdot \Gamma - \left(\frac{5}{\epsilon}\right) \cdot \Gamma.$$
\end{lemma}

We will prove the above lemma in \Cref{sec:new:lm:machine}. Given this, we can conclude the following useful observation that complements \Cref{cor:new:lm:job}. Consider a level $k$ and a machine $i$ whose current level is significantly higher than $k$, i.e., at least $k+\alpha-1$. 
When we average over machine-level nodes $(i, k),\dots,(i, \ell(i))$, each of them receives at least $2 \left(1 + \epsilon - \epsilon^2\right)\Gamma$ units of weighted flow \emph{only from} job-level nodes whose current level is at least $k-6L-1$. This is captured by the corollary below.

\thatchaphol{I would write $k$ instead of $k_i$ ($k_i$ is not determined by $i$ and so it is confusing).}

\begin{corollary}
\label{cor:edge:2}
Consider any level $k_i > 6L+\alpha$ and any machine $i \in \hat{M}_{\geq k_i}$. Then, we must have $$\frac{\sum_{k = k_i}^{\ell(i)} \  \left( \sum_{(j, k_j) \in N_{\G}(i, k) \, \cap \, \hat{\J}_{\geq k_i -6L-\alpha}} w(j, k_j) \cdot z\big( (j, k_j), (i, k) \big) \right)}{\left(\ell(i) - k + 1\right)} \geq 2(1+\epsilon-\epsilon^2) \cdot \Gamma.$$ 
\end{corollary}

\begin{proof}
Consider any level $k \in \left[k_i, \ell(i)\right]$, and focus on the machine-level node $(i, k) \in \hat{\M}_{\geq k_i}$. We define
$$\J_1(k) := \{ (j, k_j) \in N_{\G}(i, k) : k_j \geq k - 6L\} \text{ and } \J_2(k) := \{ (j, k_j) \in N_{\G}(i, k) : k_j < k - 6L\}.$$
The set of neighbors of the machine-level node $(i, k)$ is partitioned into the subsets $\J_1(k)$ and $\J_2(k)$. \Cref{cl:edge:2} implies that  $\big((j, k_j), (i, k)\big)$ is a {\em small} history edge for every job-level node $(j, k_j) \in \J_2(k)$. 
This means that $z\left((j, k_j), (i, k) \right) < \left(\epsilon/(Wmn)\right)^{14}$ for all $(j, k_j) \in \J_2(k)$. 
Next, from \Cref{assume:label} and \Cref{eq:epsilon}, we  infer that 
\begin{eqnarray*}
|\J_2(k)| & \leq & n \cdot k \leq n \cdot  10\kappa (7L+\alpha) \\
& \leq & 100 \cdot n L \kappa \alpha \\
& \leq & 100 \cdot n \cdot \left(\frac{Wmn}{\epsilon}\right)^8 \cdot \frac{5}{\epsilon^3} \\
& = & 500 \cdot  n \cdot \frac{(Wmn)^8}{\epsilon^{11}}.
\end{eqnarray*}

Accordingly, we get
\begin{eqnarray}
\sum_{(j, k_j) \in \J_2(k)} w(j, k_j) \cdot z\big( (j, k_j), (i, k) \big) & \leq & |\J_2(k)| \cdot W \cdot \left(\frac{\epsilon}{W m n}\right)^{14} \nonumber  \\ 
& \leq & 500 \cdot  n \cdot \frac{(Wmn)^8}{\epsilon^{11}} \cdot W \cdot \left(\frac{\epsilon}{W m n}\right)^{14}  \nonumber \\
& \leq & \left(\frac{Wmn}{\epsilon}\right)^{11} \cdot \left(\frac{\epsilon}{W m n}\right)^{14} \label{eq:parameter:1}\\ 
& = & \left(\frac{\epsilon}{W m n}\right)^{3} \nonumber \\
& \leq &  \epsilon^2 \Gamma. \label{eq:cor:edge:200}
\end{eqnarray}

In the above derivation, \Cref{eq:parameter:1} holds as long as $Wmn$ is larger than some absolute constant (independent of $\epsilon$), and the last step holds since $\Gamma \geq 1/n$ (see \Cref{eq:gamma:100}). 
At this point, since the set $N_{\G}(i, k)$ is partitioned into $\J_1(k)$ and $\J_2(k)$,  \Cref{eq:cor:edge:200} implies that
\begin{eqnarray*}
\sum_{(j, k_j) \in \J_1(k)} w(j, k_j) \cdot z\big( (j, k_j), (i, k) \big) & = & w_z(i, k) - \sum_{(j, k_j) \in \J_2(k)} w(j, k_j) \cdot z\big( (j, k_j), (i, k) \big) \\
& \geq & w_z(i, k) - \epsilon^2 \cdot \Gamma. 
\end{eqnarray*}
Next, observe that since $k \geq k_i$, we have $\J_1(k) \subseteq N_{\G}(i, k) \, \cap \, \hat{\J}_{\geq k_i-6L-\alpha}$, by definition. Thus, we get:
\begin{equation}
\label{eq:cor:edge:2:100}
\sum_{(j, k_j) \in N_{\G}(i, k) \, \cap \, \hat{\J}_{\geq k_i-6L-\alpha}} w(j, k_j) \cdot z\big( (j, k_j), (i, k) \big)  \geq w_z(i, k) -  \epsilon^2 \cdot \Gamma. 
\end{equation}
Summing \Cref{eq:cor:edge:2:100} over all $k \in \left[k_i, \ell(i)\right]$, we get: 
\begin{equation}
\label{eq:cor:edge:2:101}
\sum_{k = k_i}^{\ell(i)} \  \left( \sum_{(j, k_j) \in N_{\G}(i, k) \, \cap \, \hat{\J}_{\geq k_i - 6L-\alpha}} w(j, k_j) \cdot z\big( (j, k_j), (i, k) \big) \right) \geq \sum_{k = k_i}^{\ell(i)} w_z(i, k) -  \epsilon^2 \Gamma \cdot \left(\ell(i) - k_i +1\right). 
\end{equation}


Now, diving both sides of \Cref{eq:cor:edge:2:101} by $\left( \ell(i) - k_i+1\right)$ and applying \Cref{new:lm:machine}, we get:
\begin{eqnarray*}
\frac{\sum_{k = k_i}^{\ell(i)} \  \left( \sum_{(j, k_j) \in N_{\G}(i, k) \, \cap \, \hat{\J}_{\geq k_i-6L-\alpha}} w(j, k_j) \cdot z\big( (j, k_j), (i, k) \big) \right)}{\left(\ell(i) - k_i + 1\right)} & \geq & 2(1+\epsilon)\Gamma -  \epsilon^2 \Gamma - \frac{5 \Gamma}{\epsilon \left(\ell(i) - k_i + 1\right)}  \\
& \geq & 2(1+\epsilon)\Gamma -   \epsilon^2 \Gamma  -\frac{5\Gamma}{\epsilon \alpha} \\
& \geq & 2(1+\epsilon - \epsilon^2) \cdot \Gamma,
\end{eqnarray*}
where the penultimate inequality holds since $i \in \hat{M}_{\geq k_i}$ (i.e., $\ell(i) - k_i+1 \geq \alpha$) and the last inequality holds because of \Cref{eq:epsilon}. This concludes the proof.
\end{proof}

\begin{lemma}
\label{new:lm:expansion:100}
For every level $k > 6L+\alpha$, we have $w\left( \hat{\J}_{\geq k-6L-\alpha}\right) \geq \left(1+\frac{\epsilon}{2}\right) \cdot \Gamma \cdot \left| \hat{\M}_{\geq k} \right|$. 
\end{lemma}

\begin{proof} 
 \Cref{cor:new:lm:job} implies that on average a job-level node $(j, k_j) \in \hat{\J}_{\geq k-6L-\alpha}$ sends out at most $2 \left(1 + \epsilon^2/10\right) \cdot w(j) = 2 \left(1+\epsilon^2/10\right) \cdot w(j, k_j)$ units of weighted flow. On the other hand, \Cref{cor:edge:2} implies that on average a machine-level node $(i, k_i) \in \hat{\M}_{\geq k}$ receives at least $2(1+\epsilon - \epsilon^2) \cdot \Gamma$ units of weighted flow  from the job-level nodes in $\hat{\J}_{\geq k-6L-\alpha}$.  Thus,  a simple counting argument \thatchaphol{I feel this is too fast. \Cref{cor:new:lm:job,cor:edge:2} are complicated expressions. How do you conclude this exactly?} implies that $2(1+\epsilon^2/10) \cdot w(\hat{\J}_{\geq k-6L-\alpha})  \geq 2(1+\epsilon-\epsilon^2) \cdot \Gamma \cdot |\hat{\M}_{\geq k}|$, or equivalently,
\begin{eqnarray*}
w(\hat{\J}_{\geq k-6L-\alpha})  \geq \frac{\left(1+\epsilon-\epsilon^2\right)}{\left(1+\frac{\epsilon^2}{10}\right)} \cdot \Gamma \cdot |\hat{\M}_{\geq k}| \geq \left(1 +\frac{\epsilon}{2}\right) \cdot \Gamma \cdot |\hat{\M}_{\geq k}|,
\end{eqnarray*} 
where the last inequality holds since $\epsilon < 1/100$ (see \Cref{eq:epsilon}). 
\end{proof}

\begin{lemma}
\label{new:cl:expansion}
As long as $\opt_f(G, w) \leq \Gamma$,   we must have $\Gamma \cdot \left| \hat{\M}_{\geq k-L} \right| \geq  w\left(\hat{\J}_{\geq k}\right)$ for all $k > L$.
\end{lemma}

\begin{proof}
Throughout the proof, we will focus on a specific time-step $t$ such that $\opt_f^{(t)}(G, w) \leq \Gamma$. Accordingly, for ease of notation, throughout the proof we will omit the superscript $t$ from all objects.

By \Cref{obs:history-edge} and \Cref{cl:edge:1}, we have $\ell(i) \geq \ell(j) - L$ for all $(j, i) \in E$. Hence, we get:
\begin{eqnarray}
\label{eq:expansion:10}
N_G(j) & \subseteq & \hat{M}_{\geq k-L} \qquad \text{ for all jobs } j \in \hat{J}_{\geq k}. \\
\label{eq:expansion:11}
\left( \ell(i) - (k-L) +1 \right) & \geq & \left( \ell(j) - k+1\right)  \text{ for all jobs } j \in \hat{J}_{\geq k} \text{ and machines } i \in N_G(j).
\end{eqnarray}

Let $x^\star : E \rightarrow [0, 1]$ be a fractional assignment in $(G, w)$ with load at most $\Gamma$.  Such an assignment exists because  $\opt_f(G, w) \leq \Gamma$. In particular, this means that $w_{x^\star}(j) = w(j)$ for all jobs $j \in J$, and $w_{x^\star}(i) \leq \Gamma$ for all machines $i \in M$. Now, consider  any job $j \in \hat{J}_{\geq k}$. By definition, we have $\sum_{i \in N_G(j)} w(j) \cdot x^\star(j,i) = w_{x^\star}(j) = w(j)$. At this point,  \Cref{eq:expansion:11} implies that:
\begin{equation}
\label{eq:expansion:12}
\sum_{i \in N_G(j)} \left( \ell(i) - (k-L)+1\right) \cdot w(j) \cdot x^\star(j,i) \geq (\ell(j) - k+1) \cdot w(j) \ \ \text{ for all jobs } j \in \hat{J}_{\geq k}.
\end{equation}

Armed with all these observations, we now derive that:
\begin{eqnarray}
\Gamma \cdot |\hat{\M}_{k-L}| & = & \sum_{i \in \hat{M}_{\geq k-L}}  \left(\ell(i) - (k-L) + 1 \right) \cdot \Gamma \label{eq:derive:1} \\
& \geq & \sum_{i \in \hat{M}_{\geq k-L}}  \left(\ell(i) - (k-L) + 1 \right) \cdot w_{x^\star}(i) \label{eq:derive:2} \\
& = & \sum_{i \in \hat{M}_{\geq k-L}} \  \sum_{j \in N_G(i)} \left(\ell(i) - (k-L) + 1 \right) \cdot  w(j) \cdot x^\star(j,i) \label{eq:derive:3} \\
& \geq & \sum_{i \in \hat{M}_{\geq k-L}} \ \sum_{j \in N_G(i) \cap \hat{J}_{\geq k}} \left(\ell(i) - (k-L) + 1 \right) \cdot  w(j) \cdot x^\star(j,i) \nonumber \\
& = & \sum_{j \in \hat{J}_{\geq k}} \ \sum_{i \in N_G(j) \cap \hat{M}_{\geq k-L}}   \left(\ell(i) - (k-L) + 1 \right) \cdot  w(j) \cdot x^\star(j,i) \label{eq:derive:5} \\
& = & \sum_{j \in \hat{J}_{\geq k}} \ \sum_{i \in N_G(j)}   \left(\ell(i) - (k-L) + 1 \right) \cdot  w(j) \cdot x^\star(j,i) \label{eq:derive:6} \\
& \geq & \sum_{j \in \hat{J}_{\geq k}}   \left(\ell(j) - k+1 \right) \cdot  w(j) \label{eq:derive:7} \\
& = & w\left( \hat{\J}_{\geq k} \right). \label{eq:derive:8}
\end{eqnarray}

In the above derivation, \Cref{eq:derive:1} holds because of \Cref{eq:equivalence}. \Cref{eq:derive:2} holds because $x^\star$ is a fractional assignment with load at most $\Gamma$, and hence $w_{x^\star}(i) \leq \Gamma$ for all machines $i \in M$. \Cref{eq:derive:3} follows from the definition of the load $w_{x^\star}(i)$ on a machine $i \in M$ w.r.t.~$x^\star$. \Cref{eq:derive:5} follows from changing the order of the two summations. \Cref{eq:derive:6} follows from \Cref{eq:expansion:10}. \Cref{eq:derive:7} follows from \Cref{eq:expansion:12}. Finally, \Cref{eq:derive:8} holds because of \Cref{eq:equivalence}. The lemma follows from \Cref{eq:derive:8}.
\end{proof}

\begin{corollary}
\label{cor:expansion}
If $\opt_f(G, w) \leq \Gamma$, then $\left| \hat{\M}_{\geq k-(7L+\alpha)} \right| \geq (1+\epsilon/2) \cdot \left| \hat{\M}_{\geq k} \right|$ for all $k > 7L+\alpha$.
\end{corollary}

\begin{proof}
From \Cref{new:lm:expansion:100} and \Cref{new:cl:expansion}, we infer that
$$\Gamma \cdot \left| \hat{\M}_{\geq k-7L-\alpha}\right| \geq w\left(\hat{J}_{\geq k-6L-\alpha}\right) \geq \left( 1+\frac{\epsilon}{2}\right) \cdot \Gamma \cdot \left| \hat{\M}_{\geq k}\right|.$$
Dividing both sides of the resulting inequality by $\Gamma$, we get $\left| \hat{\M}_{\geq k-7L-\alpha}  \right| \geq (1+\epsilon/2) \cdot \left| \hat{\M}_{\geq k}\right|$. 
\end{proof}

\paragraph{Proof of \Cref{lm:upper:bound:level}.}  For the sake of contradiction, suppose that $\opt_f(G, w) \leq \Gamma$, and 
\begin{itemize}
\item (a) either there is a machine $i \in M$ at level $\ell(i) = 10 \kappa (7L+\alpha)$, 
\item  (b) or there is a job $j \in J$ at level $\ell(j) = 10\kappa(7L+\alpha)$. 
\end{itemize}
If we are in case (a), then $i \in \hat{M}_{\geq  10\kappa(7L+\alpha) - \alpha+1} \subseteq \hat{M}_{\geq 5\kappa(7L+\alpha)}$, and hence
\begin{equation}
\label{eq:expansion:1}
\left| \hat{\M}_{\geq 5\kappa(7L+\alpha)} \right| \geq 1.
\end{equation}

Otherwise, if we are in case (b), then \Cref{obs:history-edge} and \Cref{cl:edge:1} imply that $\ell(i) \geq \ell(j) - L$ for all machines $i \in N_G(j)$. Clearly, we have $N_G(j) \neq \emptyset$, for otherwise the input $(G, w)$ does not admit any valid assignment. Accordingly, we get 
$$\emptyset \subset N_G(j) \subseteq M_{\geq \ell(j) - L} = M_{\geq 10\kappa(7L+\alpha) - L} \subseteq \hat{M}_{\geq 10\kappa(7L+\alpha) - L - \alpha+1} \subseteq \hat{M}_{\geq 5\kappa(7L+\alpha)}.$$
This implies that $\hat{\M}_{\geq 5\kappa(7L+\alpha)} \neq \emptyset$. In other words, \Cref{eq:expansion:1} holds even under case (b).

\medskip
Now, starting from \Cref{eq:expansion:1}, we invoke \Cref{cor:expansion} for 
$\kappa$  successive iterations,  to get:
\begin{equation}
\label{eq:expansion:2}
\left| \hat{\M}_{\geq  4\kappa  (7L+\alpha)} \right| \geq (1+\epsilon/2)^{\kappa} \cdot \left| \hat{\M}_{\geq 5\kappa(7L+\alpha)} \right| \geq (1+\epsilon/2)^{\kappa}.
\end{equation}

In contrast, by \Cref{assume:label}, every machine $i \in M$ contributes at most $10\kappa(7L+\alpha)$ machine-level nodes to the set $\hat{\M}_{\geq  4\kappa  (7L+\alpha)}$. This implies that 
\begin{equation}
\label{eq:expansion:5}
\left| \hat{\M}_{\geq 4\kappa (7L+\alpha)} \right| \leq 10 m \kappa(7L+\alpha).
\end{equation}


From \Cref{eq:expansion:2}, \Cref{eq:expansion:5} and \Cref{eq:epsilon}, we derive that 
\begin{equation}
\label{eq:contradict}
10m \kappa(7L+\alpha) \geq (1+\epsilon/2)^{\kappa} = (Wmn/\epsilon)^{6}.
\end{equation}

However, \Cref{eq:epsilon} implies that \Cref{eq:contradict} cannot hold, as long as $mnW$ is larger than some absolute constant (independent of $\epsilon$). This leads to the desired contradiction.
\qed 

\subsubsection{Bounding the Total Flow out of Job-Level Nodes: Proof of \Cref{new:lm:job}}
\label{sec:new:lm:job}

Throughout the proof, we fix a time-step $t$, job $j \in J^{(t)}$ and level $0 < k_j < \ell^{(t)}(j)$. 
We will show that
$$\sum_{k = k_j}^{\ell^{(t)}(j)} z^{(t)}(j, k) \leq 2 \cdot \left(\ell^{(t)}(j) - k_j+1\right) + \left(\frac{1}{\epsilon}\right).$$
For every level $0 < k_j \leq k  \leq \ell^{(t)}(j)$, let $t_k$ denote the time-step at which the job $j$ entered level $k$. To ensure consistency of notations, we let $t_{\ell^{(t)}(j)+1} := t$. For every level $0 < k_j \leq k \leq \ell^{(t)}(j)$, we define
\begin{equation}
\label{eq:flow:partition}
z_{\gain}(j, k) := z^{(t_k)}(j, k) \text{ and } z_{\loss}(j, k) := z^{(t_{k+1})}(j, k) - z^{(t_k)}(j, k).
\end{equation}

To interpret these notations, consider the time-interval $[t_k, t_{k+1}]$, and recall that the value of $z(j, k)$ can only increase over time throughout this interval (see \Cref{obs:history-edge}). At time-step $t_k$, the job $j$ gets boosted for the $k^{th}$ time and increases its level from $k-1$ to $k$. The symbol $z_{\gain}(j, k)$ denotes the total flow sent out by the job $j$ immediately after this event. The subscript ``$\gain$'' refers to the fact that the job $j$ {\em gains} flow due to this event at time-step $t_k$. Subsequently, as long as the job $j$ continues to remain at level $k$, it can potentially {\em lose} parts of the flow it is sending out, due to one of the following two types of events. 
\begin{itemize}
    \item (i) An edge incident on  $j$ gets deleted from $G$. This  does not lead to any change in $z(j, k)$.
    \item (ii) A machine $i \in N_G(j)$ gets whacked. This leads to a {\em loss} in the flow sent out by $j$ and an increase in the value of $z(j, k)$, along with the insertion of a new history edge incident on $(j, k)$.
\end{itemize}

The symbol $z_{\loss}(j, k)$ denotes the total increase in the value of $z(j, k)$ during the time-interval $[t_k, t_{k+1}]$, due to all events of type (ii). The subscript ``$\loss$'' refers to the fact that the job $j$ {\em loses} flow during all these concerned events.

If $k < \ell^{(t)}(j)$, then it is easy to verify that $t_{k+1} \leq t$ and that the value of $z(j, k)$ remains unchanged during the time-interval $\left[t_{k+1}, t\right]$; this is because no history edge incident on $(j, k)$ gets inserted into $\G$ when the job $j$ is at level $\geq k+1$. In contrast, if $k = \ell^{(t)}(j)$, then we have defined $t_{k+1} := t$ for notational consistency. Accordingly, from \Cref{eq:flow:partition} we infer that $z^{(t)}(j, k) = z_{\gain}(j, k) + z_{\loss}(j, k)$ for all  $k_j \leq k \leq \ell^{(t)}(j)$. Summing this equality over all $k \in \left[ k_j, \ell^{(t)}(j) \right]$, we get:
\begin{equation}
\label{eq:flow:partition:3}
\sum_{k=k_j}^{\ell^{(t)}(j)} z^{(t)}(j,k) = \sum_{k=k_j}^{\ell^{(t)}(j)} z_{\gain}(j, k) + \sum_{k=k_j}^{\ell^{(t)}(j)} z_{\loss}(j, k).
\end{equation}

We will now upper bound the two summations on the right hand side of \Cref{eq:flow:partition:3}, one after another. In \Cref{cl:flow:partition:1}, we upper bound the first of these summations.

\begin{claim}
\label{cl:flow:partition:1}
We have $\sum_{k=k_j}^{\ell^{(t)}(j)} z_{\gain}(j, k) \leq \left( \ell^{(t)}(j) - k_j + 1 \right)$.
\end{claim}

\begin{proof}
Consider any level $k \in \left[k_j, \ell^{(t)}(j)\right]$. By definition, the value of $z_{\gain}(j, k)$ equals the total amount of flow sent out by the job $j$ when it first enters level $k$. So, \Cref{obs:jobs:flow:upper:bound} implies that $z_{\gain}(j, k) \leq 1$. The claim follows from summing this inequality over all $k \in \left[k_j, \ell^{(t)}(j)\right]$.
\end{proof}

It now remains to upper bound the second summation on the right hand side of \Cref{eq:flow:partition:3}. Towards this end, we define two quantities $F_{\gain}$ and $F_{\loss}$, which respectively correspond to the total gain and total loss in the flow $x(j)$ during the time-interval $\left[ t_{k_j}, t\right]$, as described below.

\begin{wrapper}
At time-step $t_{k_j}$, we initialize $F_{\gain} \leftarrow 0$ and $F_{\loss} \leftarrow 0$. Subsequently, at each time-step $t' \in \left( t_{k_j}, t\right]$, we update the values of $F_{\gain}$ and $F_{\loss}$ as per the following rules.
\begin{enumerate}
\item \label{case-a} Suppose that the job $j$ gets boosted during time-step $t'$, which implies that $x^{(t')}(j) > x^{(t'-1)}(j)$. In this case, we set $F_{\gain} \leftarrow F_{\gain} + \left(x^{(t')}(j) - x^{(t'-1)}(j) \right)$. 
\item \label{case-b} Suppose that some machine $i \in N_G(j)$ gets whacked during time-step $t'$, which implies that $x^{(t')}(j) < x^{(t'-1)}(j)$. In this case, we set $F_{\loss} \leftarrow F_{\loss} + \left(x^{(t'-1)}(j) - x^{(t')}(j) \right)$. 
\item \label{case-c} Suppose that some edge $(j, i) \in E$ incident on $j$ gets deleted during time-step $t'$, which implies that $x^{(t')}(j) < x^{(t'-1)}(j)$. In this case, we set $F_{\loss} \leftarrow F_{\loss} + \left(x^{(t'-1)}(j) - x^{(t')}(j) \right)$.
\item Otherwise, we have $x^{(t')}(j) = x^{(t'-1)}(j)$, and  $F_{\gain}, F_{\loss}$ also remain unchanged.
\end{enumerate}
Henceforth, unless explicitly specified otherwise, the symbols $F_{\loss}$ and $F_{\gain}$ will refer to their values at the end of time-step $t$.
\end{wrapper}

Next, we upper bound $F_{\loss}$ in terms of $F_{\gain}$, as stated in the claim below.

\begin{claim}
\label{cl:loss-vs-gain:1}
We have $F_{\loss} \leq F_{\gain} + 1$.
\end{claim}

\begin{proof}
From the definitions of $F_{\loss}$ and $F_{\gain}$, we conclude that $x^{(t)}(j) = x^{(t_{k_j})}(j) + F_{\gain} - F_{\loss}$, or equivalently, $F_{\loss} = F_{\gain} + \left( x^{(t_{k_j})}(j)  - x^{(t)}(j)\right)$. The claim now follows from \Cref{obs:jobs:flow:upper:bound}.
\end{proof}

Next, we relate $F_{\gain}$ to the number of times the job $j$ gets boosted during the time-interval $(t_{k_j}, t]$. 

\begin{claim}
\label{cl:loss-vs-gain:2}
We have $F_{\gain} \leq \epsilon \cdot \left( \ell^{(t)}(j) - k_j \right)$.
\end{claim}

\begin{proof}
The value of $F_{\gain}$ is initialized at $0$ when the job $j$ enters level $k_j$. Subsequently, the value of $F_{\gain}$ increases if and only if the job $j$ moves up from some level $k \in \left[k_j, \ell^{(t)}(j) \right)$ to level $k+1$. Consider the time-step $t_{k+1}$ at which such an event occurs. When the job $j$ gets boosted at time-step $t_{k+1}$, we set $x(j, i) \leftarrow (1+\epsilon) \cdot x(j, i)$ for every machine $i \in N^{(t_{k+1})}_G(j)$. Let $\Delta_{\gain}(t_{k+1})$ denote the increase in the value of $F_{\gain}$ at time-step $t_{k+1}$.  Accordingly, we get:
\begin{equation}
\Delta_{\gain}\left(t_{k+1}\right)  =  x^{(t_{k+1})}(j) - x^{(t_{k+1}-1)}(j) 
 =  \epsilon \cdot x^{(t_{k+1}-1)}(j) 
 \leq  \epsilon. \label{eq:loss-vs-gain:10}
\end{equation}

The last inequality follows from \Cref{obs:jobs:flow:upper:bound}. Since $F_{\gain} = \sum_{k = k_j}^{\ell^{(t)}(j)-1} \Delta_{\gain}\left(t_{k+1}\right)$, we obtain the claim by summing \Cref{eq:loss-vs-gain:10} over all $k \in \left[k_j, \ell^{(t)}(j) \right)$.
\end{proof}

Finally, we upper bound the sum $\sum_{k=k_j}^{\ell^{(t)}(j)} z_{\loss}(j, k)$ in terms of $F_{\loss}$.

\begin{claim}
\label{cl:loss-vs-gain:3}
We have $\sum_{k=k_j}^{\ell^{(t)}(j)} z_{\loss}(j, k) \leq (1/\epsilon) \cdot F_{\loss}$.
\end{claim}

\begin{proof}
Consider any time-step $t' \in (t_{k_j}, t]$ with $\ell^{(t')}(j) = k' \in [k_j, \ell^{(t)}(j)]$ (say), which involves the whacking of some machine $i \in N_{G}^{(t')}(j)$. During this event, we set $x(j, i) \leftarrow (1+\epsilon)^{-1} \cdot x(j, i)$, and hence $x^{(t')}(j, i) \leftarrow (1+\epsilon)^{-1} \cdot x^{(t'-1)}(j, i)$. Let $x^{(t'-1)}(j, i) = \lambda$ (say), so that  $x^{(t')}(j, i) = \lambda/(1+\epsilon)$. Due to this event, the value of $z_{\loss}(j, k')$ increases by $\lambda/(1+\epsilon)$, whereas the value of $F_{\loss}$ increases by $\lambda - \lambda/(1+\epsilon) = \epsilon \cdot \lambda/(1+\epsilon)$.  Moreover, the value of $z_{\loss}(j, k')$ increases if and only if some machine neighboring the job $j$ gets whacked. In other words, for all $k' \in [k_j, \ell^{(t)}(j)]$, every unit increase in $z_{\loss}(j, k')$ can be charged to an $\epsilon$ unit increase in $F_{\loss}$. This implies the claim.
\end{proof}

We are now ready to upper bound the second summation on the right hand side of \Cref{eq:flow:partition:3}.

\begin{corollary}
\label{cor:loss-vs-gain:4}
We have $\sum_{k=k_j}^{\ell^{(t)}(j)} z_{\loss}(j, k) \leq \left(\ell^{(t)}(j) - k_j\right) + (1/\epsilon)$.
\end{corollary}

\begin{proof}
From \Cref{cl:loss-vs-gain:1}, \Cref{cl:loss-vs-gain:2} and \Cref{cl:loss-vs-gain:3}, we infer that 
$$\sum_{k=k_j}^{\ell^{(t)}(j)} z_{\loss}(j, k) \leq (1/\epsilon) \cdot F_{\loss} \leq (1/\epsilon) \cdot F_{\gain} + (1/\epsilon) \leq \left(\ell^{(t)}(j) - k_j\right) + (1/\epsilon).$$
This concludes the proof of the corollary.
\end{proof}

At this point, \Cref{new:lm:job} follows from \Cref{eq:flow:partition:3}, \Cref{cl:flow:partition:1} and \Cref{cor:loss-vs-gain:4}. \qed

\subsubsection{Bounding the Total Weighted Flow into Machine-Level Nodes: Proof of \Cref{new:lm:machine}}
\label{sec:new:lm:machine}

Throughout the proof, we fix a time-step $t$, machine $i \in M$ and level $0 < k_i < \ell^{(t)}(i)$. We will show:
$$\sum_{k = k_i}^{\ell^{(t)}(i)} w_z^{(t)}(i, k) \geq 2 (1+\epsilon) \cdot \left(\ell^{(t)}(i) - k_i+1\right) \cdot \Gamma - \left(\frac{5}{\epsilon}\right) \cdot \Gamma.$$

We start with two useful observations which bound the total weighted flow absorbed by $i$.

\begin{observation}
\label{obs:whack:machine}
Immediately after the machine $i$ gets whacked, we must have $w_x(i) > \beta \cdot \Gamma$.
\end{observation}

\begin{proof}
The machine $i$ gets whacked only if $w_x(i) > \beta (1+\epsilon) \cdot \Gamma$ (see Line~\ref{line:whack:preprocess} in \Cref{alg:preprocess}, and Line~\ref{line:whack:update} in \Cref{alg:delete-edge}). Further, whenever the machine $i$ gets whacked, we set $x(j, i) \leftarrow x(j, i) \cdot (1+\epsilon)^{-1}$ for all $j \in N_G(i)$. Now, suppose that the machine $i$ gets whacked at a time-step (say) $t'$. Then, we must have $w_x^{(t')}(i) = w_x^{(t'-1)}(i) \cdot (1+\epsilon)^{-1} > \beta (1+\epsilon) \cdot \Gamma \cdot (1+\epsilon)^{-1} > \beta \cdot \Gamma$.
\end{proof}

\begin{observation}
\label{obs:machine:bound}
We always have $0 \leq w_x(i) \leq 3 \Gamma$.
\end{observation}

\begin{proof}
At the start of preprocessing (specifically, after Line~\ref{line:level:init:1} in \Cref{alg:preprocess}), 
we have $w_x(i) = \sum_{j \in N_G(i)} w(j) \cdot x(j,i) \leq \sum_{j \in N_G(i)} W \cdot x(j,i)  = \sum_{j \in N_G(i)} W \cdot \left(\frac{\epsilon}{Wmn}\right)^2 \leq  n  W \cdot \left(\frac{\epsilon}{Wmn}\right)^2 \leq \Gamma$,  where the last inequality holds since $\Gamma \geq 1/n$ (see \Cref{eq:gamma:100}). 

Before proceeding any further, recall that our algorithm prioritizes whacking machines over boosting jobs (see the {\bf while} loops in \Cref{alg:preprocess} and \Cref{alg:delete-edge}). In a bit more detail, as long as there exists at least one machine with load more than $\beta(1+\epsilon) \cdot \Gamma$, our algorithm does not boost any job.  Furthermore, whacking a machine $i' \neq i$ does not change the load of the machine $i$.

 Next, note that whenever a job $j \in N_G(i)$ gets boosted at (say) a time-step $t'$, the value of $w_x(i)$ increases by $w(j) \cdot \left(x^{(t')}(j, i) - x^{(t'-1)}(j,i)\right) = \epsilon \cdot w(j) \cdot x^{(t'-1)}(j, i) \leq \epsilon \cdot w_x^{(t'-1)}(i) \leq \epsilon \cdot \beta(1+\epsilon) \cdot \Gamma \leq 2\epsilon \cdot \Gamma$, where the penultimate inequality follows from the discussion in the preceding paragraph and the last inequality follows from \Cref{eq:epsilon}. 
 Similarly, whenever a job $j \in N_G(i)$ gets inserted, the value of $w_x(i)$ increases by $\left(\frac{\epsilon}{Wmn}\right)^2 \cdot (1+\epsilon)^{-\ell(i)} \leq  \epsilon^2 \Gamma$ (see Line~\ref{line:level:init:2} of \Cref{alg:delete-edge}), where the last inequality holds since $\Gamma \geq 1/n$ (see \Cref{eq:gamma:100}). 
 In contrast, deletion of an edge incident on the machine $i$ {\em decreases} the value of $w_x(i)$. To summarize, the value of $w_x(i)$ starts at $\leq \Gamma$, and increases in increments of at most an additive $\max(2\epsilon \Gamma, \epsilon^2 \Gamma) \leq \Gamma$.


According, we infer that we always have $w_x(i) \leq \beta(1+\epsilon) \cdot \Gamma + \Gamma \leq 3\Gamma$ 
(see \Cref{eq:epsilon}).
\end{proof}

For every level $0 < k_i \leq k  \leq \ell^{(t)}(i)$, let $t_k$ be the time-step at which the machine $i$ entered level $k$. To ensure consistency of notations, we let $t_{\ell^{(t)}(i)+1} := t$. For every level $0 < k_i \leq k \leq \ell^{(t)}(i)$, we define:
\begin{equation}
\label{eq:flow:partition:machine}
w_{z_{\loss}}(i, k) := w_z^{(t_k)}(i, k) \text{ and } w_{z_{\gain}}(j, k) := w_z^{(t_{k+1})}(i, k) - w_z^{(t_k)}(i, k).
\end{equation}

To interpret these notations, observe that at the time-step $t_k$, the machine $i$ gets whacked for the $k^{th}$ time and increases its level from $k-1$ to $k$. The symbol $w_{z_{\loss}}(i, k)$ denotes the total weighted flow absorbed by $i$ immediately after this event. The subscript ``$\loss$'' refers to the fact that the machine $i$ {\em loses} flow due to this event. Subsequently, as long as the machine $i$ remains at level $k$, it {\em gains} flow only if either (i) a job $j \in N_G(i)$ gets boosted or (ii) a job $j$ adjacent to the machine $i$ gets inserted. Any such event leads to an increase in the value of $w_z(i, k)$, accompanied by the insertion of a new history edge incident on $(i, k)$ in $\G$.

The symbol $w_{z_{\gain}(i, k)}$ denotes the total increase in the value of $w_z(i, k)$ during the time-interval $[t_k, t_{k+1}]$ due to all events described above.\footnote{We remark that, as long as the machine $i$ remains at level $k$, it {\em loses} flow whenever an edge $(j, i) \in E$ incident on $i$ gets deleted. However, these edge deletions do not change the value of $w_z(i, k)$ and do not influence our definition of $w_{z_{\loss}}(i,k)$.}
 The subscript ``$\gain$'' refers to the fact that the machine $i$ {\em gains} weighted flow during all these concerned events.

If $k < \ell^{(t)}(i)$, then it is easy to verify that $t_{k+1} \leq t$, and that the value of $w_z(i, k)$ remains unchanged during the time-interval $\left[t_{k+1}, t\right]$; this is because no history edge incident on $(i, k)$ gets inserted into $\G$ when the machine $i$ is at level $\geq k+1$. In contrast, if $k = \ell^{(t)}(i)$, then we have defined $t_{k+1} := t$. Accordingly, from \Cref{eq:flow:partition:machine} we infer that $w_z^{(t)}(i, k) = w_{z_{\gain}}(i, k) + w_{z_{\loss}}(i, k)$ for all  $k_i \leq k \leq \ell^{(t)}(i)$. Summing this equality over all $k \in \left[ k_i, \ell^{(t)}(i) \right]$, we get:
\begin{equation}
\label{eq:flow:partition:3:machine}
\sum_{k=k_i}^{\ell^{(t)}(i)} w_z^{(t)}(i,k) =  \sum_{k=k_i}^{\ell^{(t)}(i)} w_{z_{\loss}}(i, k) + \sum_{k=k_i}^{\ell^{(t)}(i)} w_{z_{\gain}}(i, k).
\end{equation}

We will now lower bound the two summations on the right hand side of \Cref{eq:flow:partition:3:machine}, one after another. In \Cref{cl:machine:1}, we lower bound the first of these summations.

\begin{claim}
\label{cl:machine:1}
We have $\sum_{k=k_i}^{\ell^{(t)}(i)} w_{z_{\loss}}(i, k) > \beta \left(\ell^{(t)}(i) - k_i + 1\right) \cdot \Gamma$.
\end{claim}

\begin{proof}
Consider any level $0 < k_i \leq k \leq \ell^{(t)}(i)$. \Cref{eq:flow:partition:machine} implies that $w_{z_{\loss}}(i,k)$ equals the value of $w_x(i)$ immediately after the machine $i$ gets whacked for the $k^{th}$ time. Hence, by \Cref{obs:whack:machine}, we have $w_{z_{\loss}}(i, k) > \beta \cdot \Gamma$. The claim follows from summing this inequality over all $k \in \left[k_i, \ell^{(t)}(i)\right]$.
\end{proof}

It now remains to lower bound the second summation on the right hand side of \Cref{eq:flow:partition:machine}. Towards this end, we define three quantities $W_{\gain}$, $W_{\gain}^\star$ and $W_{\loss}$, 
which respectively correspond to the total gain in the flow $w_x(i)$ during the time-interval $\left[ t_{k_i}, t\right]$ due to boosting of jobs, the total gain in the flow $w_x(i)$ during the time-interval $\left[ t_{k_i}, t\right]$ due to insertions of jobs, and the total loss in the flow $w_x(i)$ during the time-interval $\left[ t_{k_i}, t\right]$, as described below.

\begin{wrapper}
At time-step $t_{k_i}$, we initialize $W_{\gain} \leftarrow 0$, $W_{\gain}^\star \leftarrow 0$ and $W_{\loss} \leftarrow 0$. Subsequently, at each time-step $t' \in \left( t_{k_i}, t\right]$, we update the values of $W_{\gain}, W_{\gain}^\star, W_{\loss}$ as per the following rules.
\begin{enumerate}
\item \label{case-a} Suppose that the machine $i$ gets whacked during time-step $t'$, which implies that $w_x^{(t')}(i) < w_x^{(t'-1)}(i)$. In this case, we set $W_{\loss} \leftarrow W_{\loss} + \left(w_x^{(t'-1)}(i) - w_x^{(t')}(i) \right)$. 
\item \label{case-b} Suppose that some job $j \in N^{(t')}_G(i)$ gets boosted during time-step $t'$, which implies that $w_x^{(t')}(i) > w_x^{(t'-1)}(j)$. In this case, we set $W_{\gain} \leftarrow W_{\gain} + \left(w_x^{(t')}(i) - w_x^{(t'-1)}(i) \right)$. 
\item \label{case-insert} Suppose that some job $j$ gets inserted during time-step $t'$, with $(j, i) \in E^{(t')}$. This implies that $w_x^{(t')}(i) > w_x^{(t'-1)}(j)$. In this case, we set $W_{\gain}^\star \leftarrow W_{\gain}^\star + \left(w_x^{(t')}(i) - w_x^{(t'-1)}(i) \right)$. 
\item \label{case-c} Suppose that some edge $(j, i) \in E^{(t'-1)} \setminus E^{(t')}$ gets deleted during time-step $t'$, which implies that $w_x^{(t')}(i) < w_x^{(t'-1)}(i)$. In this case, we set $W_{\loss} \leftarrow W_{\loss} + \left(w_x^{(t'-1)}(i) - w_x^{(t')}(i) \right)$.
\item Otherwise, we have $w_x^{(t')}(i) = w_x^{(t'-1)}(i)$, and  $W_{\gain}, W_{\gain}^\star, W_{\loss}$ also remain unchanged.
\end{enumerate}
Henceforth, unless explicitly specified otherwise, we use the symbols $W_{\gain}, W_{\gain}^\star, W_{\loss}$ to refer to their values at the end of time-step $t$.
\end{wrapper}

Next, we lower bound $W_{\gain} + W_{\gain}^\star$ in terms of $W_{\loss}$, as stated in the claim below.

\begin{claim}
\label{cl:loss-vs-gain:100}
We have $W_{\gain} + W_{\gain}^\star \geq W_{\loss} - 3 \Gamma$.
\end{claim}

\begin{proof}
From the definitions of $W_{\loss}, W_{\gain}$ and $W_{\gain}^\star$, we conclude that $w_x^{(t)}(i) = w_x^{(t_{k_i})}(i) + W_{\gain} + W_{\gain}^\star - W_{\loss}$, or equivalently, $W_{\gain} + W_{\gain}^\star = W_{\loss} + \left( w_x^{(t)}(i)  - w_x^{(t_{k_i})}(i)\right)$. 
The claim now follows from \Cref{obs:machine:bound}.
\end{proof}

Moving forward, we show that the value of $W_{\gain}^\star$ is negligibly small.

\begin{claim}
\label{cl:insert:200}
We have $W^\star_{\gain} \leq \epsilon^2 \Gamma$.
\end{claim}

\begin{proof}
At time-step $t_{k_i}$, we have $W_{\gain}^\star = 0$. Subsequently, during the time-interval $(t_{k_i}, t]$, whenever a job $j$ adjacent to the machine $i$ gets inserted, the value of $W_{\gain}^\star$ increases by $\left(\frac{\epsilon}{Wmn}\right)^2 \cdot (1+\epsilon)^{-\ell(i)} \leq \left(\frac{\epsilon}{Wmn}\right)^2$ (see Line~\ref{line:level:init:2} of \Cref{alg:delete-edge}).  Since there can be at most $n$ job insertions during the time-interval $(t_{k_i}, t]$, at time-step $t$ we must have $W_{\gain}^\star \leq n \cdot \left(\epsilon/(Wmn)\right)^2 \leq \epsilon^2/n \leq \epsilon^2 \Gamma$, where the last inequality holds since $\Gamma \geq 1/n$ (see \Cref{eq:gamma:100}). 
\end{proof}

\begin{corollary}
\label{cor:loss-vs-gain:100}
We have $W_{\gain} \geq W_{\loss} - (3+\epsilon^2) \Gamma$.
\end{corollary}

\begin{proof}
Follows from \Cref{cl:loss-vs-gain:100} and \Cref{cl:insert:200}.
\end{proof}

Next, we relate $W_{\loss}$ to the number of times $i$ gets whacked during the time-interval $(t_{k_i}, t]$. 

\begin{claim}
\label{cl:loss-vs-gain:101}
We have $W_{\loss} \geq \epsilon \Gamma \cdot \left( \ell^{(t)}(i) - k_i \right)$.
\end{claim}

\begin{proof}
The value of $W_{\loss}$ is initialized at $0$ when the machine $i$ enters level $k_i$. Subsequently, whenever the machine $i$ moves up from some level $k \in \left[k_i, \ell^{(t)}(i) \right)$ to level $k+1$, the value of $W_{\loss}$ increases. Consider the time-step $t_{k+1}$ at which such an event occurs. When the machine $i$ gets whacked at time-step $t_{k+1}$, we set $x(j, i) \leftarrow (1+\epsilon)^{-1} \cdot x(j, i)$ for every job $j \in N_G^{(t_{k+1})}(i)$. This implies that $w_x^{(t_{k+1})}(i) = (1+\epsilon)^{-1} \cdot w_x^{(t_{k+1}-1)}(i)$. Furthermore, by \Cref{obs:whack:machine}, we have $w_x^{(t_{k+1})}(i) > \beta \Gamma > \Gamma$.

Let $\Delta_{\loss}(t_{k+1})$ denote the increase in the value of $W_{\loss}$, which occurs due to the machine $i$ getting whacked at time-step $t_{k+1}$.  From the discussion in the preceding paragraph, we infer that:
\begin{equation}
\Delta_{\loss}\left(t_{k+1}\right)  =  w_x^{(t_{k+1})}(i) - w_x^{(t_{k+1}-1)}(i) 
 =  \epsilon \cdot w_x^{(t_{k+1})}(i) 
 \geq  \epsilon \Gamma. \label{eq:loss-vs-gain:200}
\end{equation}

Since $W_{\loss} \geq \sum_{k = k_i}^{\ell^{(t)}(i)-1} \Delta_{\loss}\left(t_{k+1}\right)$,  the claim follows from \Cref{eq:loss-vs-gain:200}.
\end{proof}

Finally, we lower bound  $\sum_{k=k_i}^{\ell^{(t)}(i)} w_{z_{\gain}}(j, k)$ in terms of $W_{\gain}$.

\begin{claim}
\label{cl:loss-vs-gain:300}
We have $\sum_{k=k_i}^{\ell^{(t)}(i)} w_{z_{\gain}}(i, k) \geq \frac{(1+\epsilon)}{\epsilon} \cdot W_{\gain}$.
\end{claim}

\begin{proof}
Consider any time-step $t' \in (t_{k_i}, t]$ with $\ell^{(t')}(i) = k' \in [k_i, \ell^{(t)}(i)]$ (say), which involves the boosting of some job $j \in N_{G}^{(t')}(i)$. During this event, we set $x(j, i) \leftarrow (1+\epsilon) \cdot x(j, i)$. Let $w(i) \cdot x^{(t'-1)}(j, i) = \lambda$ (say), so that  $w(i) \cdot x^{(t')}(j, i) = (1+\epsilon) \lambda$. Due to this event, the value of $w_{z_{\gain}}(j, k')$ increases by $(1+\epsilon) \lambda$, whereas the value of $W_{\gain}$ increases by $(1+\epsilon)\lambda - \lambda = \epsilon \lambda$.  Moreover, the value of $W_{\gain}$ increases if and only if some job neighboring the machine $i$ gets boosted. In other words, every unit increase in the value of $W_{\gain}$ can be charged to a $(1+\epsilon)/\epsilon$ unit increase in the value of $w_{z_{\gain}}(i, k')$,  for some $k' \in [k_j, \ell^{(t)}(i)]$. This implies the claim.
\end{proof}

We are now ready to lower bound the second summation on the right hand side of \Cref{eq:flow:partition:3:machine}.

\begin{corollary}
\label{cor:loss-vs-gain:400}
We have $\sum_{k=k_i}^{\ell^{(t)}(i)} w_{z_{\gain}}(i, k) \geq (1+\epsilon) \cdot \left(\ell^{(t)}(i) - k_i + 1\right) \cdot \Gamma  - (5/\epsilon) \cdot \Gamma$.
\end{corollary}

\begin{proof}
From \Cref{cor:loss-vs-gain:100}, \Cref{cl:loss-vs-gain:101} and \Cref{cl:loss-vs-gain:300}, we infer that 
\begin{eqnarray*}
\sum_{k=k_i}^{\ell^{(t)}(i)} w_{z_{\gain}}(i, k) & \geq & \frac{(1+\epsilon)}{\epsilon} \cdot W_{\gain} \\ 
& \geq &  \frac{(1+\epsilon)}{\epsilon} \cdot \left( W_{\loss} -  (3+\epsilon^2) \cdot \Gamma \right) \\
& \geq &  \frac{(1+\epsilon)}{\epsilon} \cdot  W_{\loss} -  \frac{4}{\epsilon} \cdot  \Gamma  \\
& \geq & (1+\epsilon) \cdot \Gamma \cdot \left( \ell^{(t)}(i) -  k_i \right) - \frac{4}{\epsilon} \cdot  \Gamma \\
& \geq &  (1+\epsilon) \cdot \Gamma \cdot \left( \ell^{(t)}(i) -  k_i + 1 \right) - \frac{5}{\epsilon} \cdot  \Gamma, 
\end{eqnarray*}
where the third and the last inequality hold because $\epsilon < 1/100$ (see \Cref{eq:epsilon}). 
\end{proof}

At this point, \Cref{new:lm:machine} follows from \Cref{cl:machine:1}, \Cref{cor:loss-vs-gain:400} and \Cref{eq:epsilon}. 
\qed

%% file: Rounding.tex
\section{Fully Dynamic Rounding for Load Balancing}

\label{sec:rounding}

In this section, we present our algorithm for rounding a dynamic $\epsilon$-fractional assignment. We restate the theorem summarizing the guarantees of our rounding framework for convenience.

\thmrounding*



Our rounding framework consists of two components. First, we show that the online load balancing rounding framework of Krishnaswamy et al.~\cite{Krishnaswamy0S23} can be seamlessly extended to the dynamic setting using folklore dynamic data structures. Their algorithm, for a given load balancing instance $(G = (J \cup M,E),w)$ and an $\epsilon$-fractional assignment $x$, maintains a dynamic graph $H = (J \cup S, E_H)$ between jobs and \emph{slots}. We refer to such a structure as a \emph{slot-decomposition}, following the terminology of Shmoys and Tardos \cite{ShmoysT93}; it is formally defined in \Cref{def:slot-decomposition}.

In $H$, each machine $i \in M$ is represented by a set of vertices $S(i)$, which we refer to as \emph{slots}. Let $\mu^*$ be a job-perfect matching of $H$, that is, a matching covering all jobs in $J$. Consider the integral assignment that assigns every job $j \in J$ to the machine associated with the slot matched to $j$ in $\mu^*$. The data structure ensures that such an assignment has low maximum load. Furthermore, the graph $H$ is robust under updates to $x$.

In the dynamic setting, we are not able to explicitly maintain $H$. Instead, we maintain an implicit representation of $H$. This representation explicitly maintains the nodes of $H$ and, for any job–machine pair $j,i$, can efficiently return the set of slots in $S(i)$ to which $j$ is adjacent in $H$. 

The second component of our framework enables us to efficiently maintain a job-perfect matching of $H$ as it undergoes updates. To achieve this, we exploit the fact that the fractional assignment $x$, when projected onto the edges of $H$, assigns weight at least $1-\epsilon$ to nodes in $J$ and at most $1-2\epsilon$ to nodes in $S$. \Cref{lem:slot-decomposition:matching} shows that, given access to such a fractional matching, we can match any unmatched node in $J$ via a random augmenting walk defined by $x$ in $O(\log (n)/\epsilon)$ steps. The following definitions are used to formalize the slot-decomposition framework of Krishnaswamy et al. \cite{Krishnaswamy0S23}.

\begin{definition}

For an ordered set $S = \{s_1, \dots, s_k\}$ of vertices with weights $w:S \rightarrow \mathbb{R}$, the interval of $s_i \ | \ i\in [k]$ in $S$ with respect to $w$ is defined as

$$I(S,w,s_i) = [\sum_{i' < i} w(s_{i'}), \sum_{i' \leq i} w(s_{i'})].$$ 

\end{definition}

For a fractional assignment $x$ of a load balancing instance $(G = (J \cup M,E), w)$ and a machine $i \in M$, let $x[i]: N_G(i) \rightarrow [0,1]$ denote the weight function that assigns weight $x(j,i)$ to each job $j \in N_G(i)$.

\begin{definition}

\label{def:slot-decomposition}

A slot-decomposition of a load balancing instance $(G = (J \cup M,E), w)$ with a corresponding $\epsilon$-fractional assignment $x$ is an assignment of sets of ordered vertices $S(i)$ (to which we will refer to as slots) to the machines $i \in M$ and an associated weight function $w_S:S \rightarrow [0,1]$ satisfying that $\sum_{s \in S(i)} w(s) \geq x(i)$ for all $i \in M$.

The graph induced by the slot-decomposition is a bipartite graph $H = (J \cup S, E_H)$ between jobs $J$ and slots $S = \cup_{i \in M}S(i)$. There exists an edge between slot $s \in S(i)$ and job $j \in J$ in $E_H$ if and only if $I(S(i),w_S,s) \cap I(N_G(i),x[i],j) \neq \emptyset$, where $N_G(i)$ is ordered non-increasingly with respect to $w$.

Define the cost $c(s)$ of a slot $s \in S(i)$ of a slot-decomposition to be $W$ if it is the first slot of $S(i)$ and otherwise to be the weight of its largest-weight neighbor in $J$ (with respect to $w$) in the induced graph $H$. We call the slot-decomposition $\epsilon$-approximate if the following conditions hold:
\begin{itemize}
\item $w_S(s) \leq 1 - 2\epsilon$ for all slots $s \in S$,
\item $\sum_{s \in S(i)} c(s) \leq (1+5\epsilon) w_x(i) + W$ for all $i \in M$.
\end{itemize}




\end{definition}

The following observation shows that a left-perfect matching of the graph induced by an $\epsilon$-approximate slot-decomposition induces an integral assignment with low maximum load in the corresponding load balancing instance. 

\begin{observation}
\label{obs:slot-decomposition:approximation}

Let $x$ be an $\epsilon$-fractional assignment to a load balancing instance $(G = (J \cup M,E),w)$ with maximum load $\Gamma$ and maximum job weight $W$. Assume $H$ is the graph induced by an $\epsilon$-approximate slot-decomposition of $x$. Let $\mu^*_H$ be a job-perfect matching of $H$. Then the assignment that assigns each job in $J$ to the machine associated with its matched slot in $H$ has load at most $\Gamma(1+5\epsilon) + W$.

\end{observation}

\begin{proof}
Consider the load on any machine $i \in M$ under such an assignment. For machine $i \in M$ let $S(i,l)$ stand for the $l$-th slot of $S(i)$ according to the ordering of $S(i)$. Slots $S(i,l)$ for $l>1$ may be matched in $\mu^*_H$ only to jobs with weight at most $c(S(i,l))$. Slot $S(i,1)$ may be matched to a job with weight at most $W$. Hence, the load of $i$ is at most $(1+5\epsilon)w_x(i) + W \leq (1+5\epsilon)\Gamma + W$.
\end{proof}





For a dynamic fractional assignment $x$, let $\Delta(x,l_1)$ stand for the total change of $x$ in terms of $l_1$ distance throughout a series of updates. The following lemma states that there exists a near-optimal algorithm for maintaining an $\epsilon$-approximate slot decomposition of a fractional assignment of a dynamic load-balancing instance. Its proof is deferred to \Cref{sec:proof:dynamic-slot}.

\begin{lemma}

\label{lem:slot-decomposition:dynamic}

Let $x$ be a fractional solution to a dynamic load balancing instance $(G = (L \cup R,E),w)$ undergoing edge and vertex updates. There exists a deterministic dynamic algorithm for maintaining an $\epsilon$-approximate dynamic slot-decomposition $(S,w_S)$ of $x$ with $\tilde{O}(1/\epsilon^2)$ amortized update time per edge update to $x$. The graph $H$ induced by the decomposition satisfies the following properties:

\begin{itemize}
    \item $|N_H(j) \cap S(i)| = O(1)$ for any pair $j \in J, i\in M$,
    \item An update to edge $(j,i)$ in $x$ may only change edges of $H$ incident to $j$, or edges incident to slots whose weights change or which are inserted into or deleted from the decomposition due to the update.
\end{itemize}

The algorithm allows for queries of the form $(j,i)$ returning the set $N_H(j) \cap S(i)$ and 

$$|I(N_H(j),x[i],j) \cap I(S(i),w_H,s)|$$

(i.e. the lengths of the corresponding intersections) for every slot $s \in N_H(j) \cap S(i)$ in $\tilde{O}(1)$ time. The total number of slots undergoing insertions and deletions in the underlying decomposition is $\tilde{O}(\Delta(x,l_1))$ throughout the update sequence.

\end{lemma}

The following lemma allows us to use the data structure described in \Cref{lem:slot-decomposition:dynamic} to maintain a left-perfect matching of the induced job-slot graph $H$ efficiently. Its proof is deferred to \Cref{sec:rounding:matching}.

\begin{lemma}

\label{lem:slot-decomposition:matching}

Let $G = (L \cup R, E)$ be a bipartite graph undergoing $m$ edge updates and $n$ vertex updates, where each update is either an insertion or a deletion. Let $x$ be a dynamic fractional matching such that, at all times, $x(l) \ge 1 - \epsilon$ for all $l \in L$ and $x(r) \le 1 - 2\epsilon$ for all $r \in R$.

There exists a randomized algorithm which, given access to an oracle $\mathcal{O}$ that, on query $l \in L$, returns a neighbor $r \in N_E(l)$ with probability $x(l,r)/x(l)$, maintains an $L$-perfect matching in $G$ in $\tilde{O}(m/\epsilon)$ total update time and $\tilde{O}((D+n)/\epsilon)$ queries to $\mathcal{O}$ with high probability, where $D$ denotes the total number of deletions of edges from the maintained matching. The algorithm functions against an adaptive adversary.

\end{lemma}

We are now ready to prove \Cref{thm:rounding}. By \Cref{lem:slot-decomposition:dynamic}, we can maintain an $\epsilon$-approximate slot-decomposition of a dynamic fractional assignment in $\tilde{O}((Q+U)/\epsilon^2)$ total update time. We will use \Cref{lem:slot-decomposition:matching} to maintain a job-perfect matching $M_H^*$ in the bipartite graph $H$ induced by this decomposition. At time $t$, \Cref{obs:slot-decomposition:approximation}, applied with the current maximum job weight $W_t$, shows that this matching corresponds to an integral assignment of load at most $(1+5\epsilon)\Gamma_t+W_t$.

In order to do this, we first argue that \Cref{lem:slot-decomposition:dynamic} offers the query access required by \Cref{lem:slot-decomposition:matching} to maintain a job-perfect matching of $H$. Consider the fractional matching $y : E_H \rightarrow [0,1]$ defined as follows for any pair $j \in J$ and $s \in S(i)$, where $i \in M$:

$$y(j,s) = |I(N_H(j),x[i],j) \cap I(S(i),w_H,s)|$$

By the properties of \Cref{lem:slot-decomposition:dynamic}, we have that $y(j) \geq 1-\epsilon$ for any $j \in J$ and $y(s) \leq 1-2\epsilon$ for any slot $s \in S$. For a given job $j\in J$, we can sample neighboring slots in two steps. First, we sample a machine $i \in N_G(j)$ such that any such machine is sampled with probability $x(j,i)/x(j)$. To do this in $\tilde{O}(1)$ time, we can maintain dynamic samplers over the set $N_G(j)$ with weights $x(j,\cdot)$ in $\tilde{O}(1)$ update time using folklore data structures.

Once machine $i \in N_G(j)$ is sampled, we use the data structure of \Cref{lem:slot-decomposition:dynamic} to find the set $N_H(j) \cap S(i)$ as well as the weight $y(j,s)$ for each $s \in N_H(j) \cap S(i)$. We then sample slot $s \in N_H(j) \cap S(i)$ with probability $y(j,s)/x(j,i)$. Note that, as $|N_H(j) \cap S(i)| = O(1)$, this takes $O(1)$ time. The probability that we have sampled a specific slot works out to be: 

$$\left(\frac{x(j,i)}{x(j)}\right)\cdot \left(\frac{y(j,s)}{x(j,i)}\right) = \frac{y(j,s)}{x(j)} = \frac{y(j,s)}{y(j)}$$
as required.

Due to an update to $x$, at most $\tilde{O}(1/\epsilon^2)$ slots change in $H$ by \Cref{lem:slot-decomposition:dynamic}. The same lemma guarantees that only jobs adjacent to changed slots can have their neighborhoods in $H$ modified. Thus, one update to $x$ can unmatch at most $\tilde{O}(1/\epsilon^2)$ jobs in $M_H^*$. By \Cref{lem:slot-decomposition:matching}, the total update time is $\tilde{O}((Q+U)/\epsilon^3)$ with high probability.

\begin{observation}

For our specific application, we know that $\Delta(x,l_1) = \tilde{O}(n)$. This implies that only $\tilde{O}(n)$ jobs get unmatched in $M_H^*$ due to updates to $S,w_H$. However, as any edge deletion may unmatch an additional single job in $M_H^*$, this speedup gets absorbed in the notation. This would not be the case if the underlying dynamic load balancing problem were incremental.

\end{observation}

\subsection{Proof of \Cref{lem:slot-decomposition:dynamic}}

\label{sec:proof:dynamic-slot}

Fix a machine $i \in M$. We describe how the algorithm maintains the relevant information about $H$ in the neighborhood of $S(i)$ and the slots and weights of $S(i)$. The algorithm is identical from the perspective of all machines.

\begin{definition}
An interval-segment tree is a dynamic data structure that maintains two ordered sets $A = \{a_1,\dots,a_k\}$ and $B = \{b_1,\dots,b_l\}$ with weights $w_A: A \rightarrow \mathbb{R}^{\geq 0}$ and $w_B: B \rightarrow \mathbb{R}^{\geq 0}$. It supports updates in which elements may be inserted into or deleted from $A$ and $B$. 

Given a query element $b \in B$, it returns the elements of $A$ whose intervals in $A$ with respect to $w_A$ intersect the interval of $b$ in $B$ with respect to $w_B$, together with the lengths of these intersections.

\end{definition}

An interval-segment tree for sets with at most $n$ elements can be implemented with $O(\log n)$ update time and $O(\log n + k)$ query time, where $k$ corresponds to the number of elements returned by the query, using a folklore interval-tree data structure. For each machine $i \in M$, the algorithm maintains an interval-segment tree between the two sets $S(i)$ and $N_G(j)$ with weight functions $w_S$ and $x[i]$. Here $N_G(i)$ is ordered non-increasingly with respect to $x[i]$. 

The algorithm makes sure that $x(i) = \sum_{s\in S(i)}w_S(s)$ once an update has been processed. The algorithm partitions the slots of $S(i)$ into consecutive sets $\mathcal{B}_1,\dots,\mathcal{B}_k$, which we call \emph{buckets}. With a slight overload of notation, we use $w_S$ to define the weights of these buckets. The algorithm ensures that $w_S(\mathcal{B}) = \sum_{s\in \mathcal{B}}w_S(s)$ for any bucket $\mathcal{B} \subseteq S(i)$ once an update has been processed. 

If $x(i) \leq 1/\epsilon$ then $S(i)$ consists of a single bucket of slots. Otherwise, the algorithm ensures that every bucket of $S(i)$ has weight in $[1/\epsilon,4/\epsilon]$. Each bucket will consist of $O(1/\epsilon)$ slots. In the ordering of $S(i)$ all but the last slot of each bucket will have weight in $[1-4\epsilon,1-2\epsilon]$. The last slot of the bucket can have weight in $[0,1-2\epsilon]$.

\textbf{Resetting a Bucket. } Throughout the maintenance of the slots of $S(i)$, whenever a bucket's structure changes by a certain amount, the algorithm resets the bucket. When bucket $\mathcal{B}$ is reset, its slots are replaced by $\left\lfloor w_S(\mathcal{B}_i)/(1-3\epsilon) \right\rfloor$ slots of weight $1-3\epsilon$ and a single slot of weight $w_S(\beta) -(1-3\epsilon)\left\lfloor w(\mathcal{B})/(1-3\epsilon) \right\rfloor$ (the last unique slot is inserted only if it would have weight greater than $0$). Observe that, during a call to \Cref{alg:reset-bucket}, the weight of the bucket does not change.

\begin{algorithm}[H]
\caption{$\resetbucket (i,\mathcal{B})$ \label{alg:reset-bucket}}
  \DontPrintSemicolon
  \For{$s \in \mathcal{B}$}{
        Delete $s$ from $S(i)$ and $\mathcal{B}$\;
  }
  \For{$i \in [\left\lfloor w_S(\mathcal{B})/(1-3\epsilon) \right\rfloor]$}{
        Insert a slot with weight $1-3\epsilon$ into $\mathcal{B}$
  }
  \If{$(1-3\epsilon)\left\lfloor w_S(\mathcal{B})/(1-3\epsilon) \right\rfloor \neq w_S(\mathcal{B}) $}{
        Insert a slot of weight $w_S(\mathcal{B})-(1-3\epsilon)\left\lfloor w_S(\mathcal{B})/(1-3\epsilon) \right\rfloor$ as the last slot of $\mathcal{B}$
  }
\end{algorithm}

\textbf{Increase of $x(j,i)$.} Assume an edge $(j,i) \in E$ increases in weight under $x$ by some amount $\Delta$. If its weight was previously $0$ (or the edge was not present in the graph), the algorithm first inserts $j$ into the interval-segment tree of $m$ with weight $0$. By querying the segment-tree data structure, the algorithm finds the left-most slot $s$ neighboring $j$. The algorithm then increases the weight of $j$ in the interval-segment tree while increasing the weight of $s$ by the same amount until both have been increased by $\Delta$. Throughout this increment, the following two events may occur (possibly multiple times).

Firstly, $s$ can reach weight $1-2\epsilon$. In this case, the algorithm resets the bucket containing $s$. Second, the bucket containing $s$ can reach weight $4/\epsilon$. In this case, the algorithm splits this bucket into two, each of weight $2/\epsilon$, then resets these buckets.

\begin{algorithm}[H]
\caption{$\increaseslot (j,i,\Delta)$ \label{alg:increase-slot}}
  \DontPrintSemicolon
  \If{$x(j,i)$ was $0$ before the update}{
    Add $j$ into the interval segment-tree of $i$ with weight $0$
  }
  \While{$\Delta > 0$}{
      $s\leftarrow$ leftmost slot in $S(i) \cap N_H(j)$ \;
      $\mathcal{B} \leftarrow$ bucket of $s$ in $S(i)$\;
      $\delta = \min\{\Delta,1-2\epsilon-w_S(s), 4/\epsilon-w_S(\mathcal{B})\}$\;
      $\Delta \leftarrow \Delta - \delta$ \;
      $w_S(s) \leftarrow w_S(s) + \delta$ \;
      $w_S(\mathcal{B}) \leftarrow w_S(\mathcal{B}) + \delta$ \;
      Increase the weight of $j$ by $\delta$ in the interval segment-tree of $i$ \;
      \If{$w_S(s) = 1-2\epsilon$}{
            $\resetbucket(i,\mathcal{B})$
      }
      \If{$w_S(\mathcal{B}) = 4/\epsilon$}{
            Replace $\mathcal{B}$ with buckets $\mathcal{B}'$,$\mathcal{B}''$ such that $w_S(\mathcal{B}') = w_S(\mathcal{B}'')= 2/\epsilon$ \;
            $\resetbucket(i,\mathcal{B'})$, $\resetbucket(i,\mathcal{B}'')$ 
      }
  }
\end{algorithm}

\textbf{Decrease of $x(j,i)$.} In the opposite case, assume an edge $(j,i)$ decreases in weight under $x$ by some $\Delta$. The algorithm first queries the slots neighboring $j$ in $H$ using the interval-segment tree. Let $s$ be the leftmost slot neighboring $j$ in $N_G(j) \cap S(i)$. The algorithm reduces the weight of $j$ in the interval-segment tree and the weight of $s$ in $w_S$ at the same rate until a reduction of $\Delta$ has been applied. The following cases may occur.

If $s$ is the last slot of its corresponding bucket, it could shrink to weight $0$. In this case, the algorithm deletes it (from $S(i)$, the containing bucket, and the interval-segment tree of $i$) and continues with the left neighbor of $s$. Otherwise, if the weight of slot $s$ reduces to $1-4\epsilon$, the algorithm resets the bucket of $s$.

It could also happen that the bucket of $s$ reduces to weight $1/\epsilon$. If it is the only bucket of $i$, the algorithm does not act on this. Otherwise, it merges this bucket with one of its neighboring buckets. If the resulting bucket has weight greater than $3/\epsilon$, it splits the bucket into two equal-weight buckets (of weights in $[1.5/\epsilon,2.5/\epsilon]$), which are then reset. Finally, if $x(j,i)=0$ after the update, the algorithm deletes $j$ from the interval-segment tree.

\begin{algorithm}[H]
\caption{$\decreaseslot (j,i,\Delta)$ \label{alg:decrease-slot}}
  \DontPrintSemicolon

  \While{$\Delta > 0$}{
      $s\leftarrow$ leftmost slot in $S(i) \cap N_H(j)$ \;
      $\mathcal{B} \leftarrow$ bucket of $s$ in $S(i)$\;
      \If{$s$ is the last slot of $\mathcal{B}$}{
         $\delta = \min\{\Delta,w_S(s), w_S(\mathcal{B})-1/\epsilon\}$\;
      }
      \Else{
         $\delta = \min\{\Delta,w_S(s) - (1-4\epsilon), w_S(\mathcal{B})-1/\epsilon\}$
      }
      $w_S(s) \leftarrow w_S(S) - \delta$ \;
      $\Delta \leftarrow \Delta - \delta$ \;
      $w_S(\mathcal{B}) \leftarrow w_S(\mathcal{B}) - \delta$ \;
      Decrease the weight of $j$ by $\delta$ in the interval segment-tree of $i$ \;
      \If{$w_S(s) = 1-4\epsilon$ and $s$ is the last slot of $\mathcal{B}$}{
            $\resetbucket(i,\mathcal{B})$
      }
      \ElseIf{$w_S(s) = 0$}{
        Delete $s$ from $S(i)$ \;
      }
      \If{$w_S(\mathcal{B}) \leq 1/\epsilon$ and $\mathcal{B}$ is not the only bucket of $S(i)$}{
        $\mathcal{B}' \leftarrow $ arbitrary neighboring bucket of $\mathcal{B}$ in $S(i)$ \;
        \If{$w_S(\mathcal{B}) + w_S(\mathcal{B'}) \geq 3/\epsilon$}{
            $w_S(\mathcal{B}),w_S(\mathcal{B}) \leftarrow (w_S(\mathcal{B}) + w_S(\mathcal{B'}))/2$ \;
            $\resetbucket(i,\mathcal{B}),\resetbucket(i,\mathcal{B}')$ \;
            }
        \Else{
            $w_S(\mathcal{B}) \leftarrow w_S(\mathcal{B}) + w_S(\mathcal{B'})$ \;
            $\resetbucket(i,\mathcal{B})$ \;
            Delete $\mathcal{B}'$ and all its slots from $S(i)$
        }
      }
  }
  \If{$x(j,i)$ is $0$ after the update}{
    Delete $j$ from the interval segment-tree of $i$
  }
\end{algorithm}

\textbf{Correctness of the Algorithm.} We will prove the correctness of the algorithm through the following series of claims.

\begin{claim}

\label{cl:dynamic-slot:decomposition}

Once an update has been processed, $\sum_{s \in S(i)} w_S(s) = x(i)$, and for any bucket $\mathcal{B}$ of $S(i)$, we have $w_S(\mathcal{B}) = \sum_{s \in \mathcal{B}} w_S(s)$.

\end{claim}
\begin{proof}

Let $\mathcal{B}_i$ stand for the set of buckets of $i$. Define a sub-update to be one iteration of the while loop in either \Cref{alg:decrease-slot} or \Cref{alg:increase-slot}, and suppose that after such a sub-update the weight of $x(j,i)$ is $x(j,i)-\Delta$ or $x(j,i)+\Delta$ for weight increases and decreases, respectively. We will argue that the claim holds after each sub-update.

Unless the affected bucket is split or merged with other buckets, its weight changes at the same rate as $\Delta$. If it is split or merged with other buckets, the total weight of all buckets remains unchanged during the split and merge operations. After any split or merge operation, every affected bucket is reset, guaranteeing that its weight equals the sum of the weights of its slots. 

If, during a sub-update, a bucket is not split or merged into other buckets, then one of its contained slots changes weight at the same rate as $\Delta$. Therefore, the claim holds after every sub-update.
    
\end{proof}

Note that, by definition, the algorithm may query $N_H(j) \cap S(i)$ in $\tilde{O}(|N_H(j) \cap S(i)|)$ time using the interval segment-tree data structure. At the same time, it can find $|I(N_G(i),x[i],j) \cap I(S(i),w_H,s)|$ for every returned slot.

\begin{claim}

\label{cl:dynamic-slot:approximation}

The slot decomposition maintained by the algorithm is $\epsilon$-approximate. Furthermore, $|N_H(j) \cap S(i)| = O(1)$ for any $j \in J,i \in M$.

\end{claim}

\begin{proof}

Consider all the slots that are not the first ones in their buckets. For such a slot $s$, the slot to its left, $s^-$, is of weight at least $1-4\epsilon$, and all jobs assigned to it are of cost at least $c(s)$. Hence, the edges in $H$ incident to $s^-$ contribute at least $(1-4\epsilon) c(s)$ to $w_x(i)$. 

For the left-most slot, $c(s) = W$. Any other left-most slot $s$ of a bucket satisfies that the bucket to the left of it is of weight at least $1/\epsilon$ and that all jobs neighboring its slots in $H$ have cost at least $c(s)$. Hence, the edges incident on the bucket to its left contribute $c(s)/\epsilon$ to $w_x(i)$. This means $\sum_{s \in S(r)} c(s) \leq w_x(i)(1/(1-4\epsilon)+\epsilon) \leq w_x(i)(1+5\epsilon)$ for $\epsilon < 1/4$.

As at most one bucket may have weight smaller than $1/\epsilon$, a job neighbors $O(1)$ buckets. Since every bucket can have one slot with weight smaller than $1-4\epsilon$, this implies that jobs neighbor $O(1)$ slots.
    
\end{proof}

\begin{claim}

\label{cl:dynamic-slot:stability}

Due to an edge update to $x$, any job whose neighborhood changes in $H$ is either incident to the affected edge in $G$ or adjacent to a slot in $H$ that underwent an update (a weight change in $w_S$, an insertion, or a deletion).

\end{claim}

\begin{proof}

Consider how the intervals of a job $j$ in $N_G(i)$ with respect to $x[i]$ and of a slot $s \in S(i)$ with respect to $w_S$ evolve as the weight of another job-slot pair is adjusted. As the job weights in the interval segment-tree and the slot weights in $w_S$ are updated at the same rate, these two intervals shift by the same amount and do not change in length. Therefore, any potential intersection of these two intervals remains unchanged.

Whenever a bucket is reset, split, or merged, the resulting bucket or buckets have the same total weight under $w_S$. In addition, before and after such a reset, the weights of the slots of the affected buckets add up to the weights of the affected buckets. Therefore, slot $s$'s interval in $S(i)$ under $w_S$ does not change when a bucket not containing it is reset. This means that the neighborhood of slot $s$ in $H$ does not change due to a bucket not containing it being reset.

\end{proof}

\begin{claim}
\label{cl:dynamic-slot:update-time}

Assume that, throughout the update sequence, $x$ undergoes $U$ edge updates. Then the update time of the algorithm is $\tilde{O}(U/\epsilon^2)$. In addition, at most $\Delta(x,l_1)$ slots undergo updates (either in terms of weight in $w_S$, deletion, or insertion).

\end{claim}

\begin{proof}

A bucket may contain $O(1/\epsilon)$ slots at any time. Therefore, resetting or merging buckets results in $O(1/\epsilon)$ slots undergoing insertions or deletions. This can be implemented in $\tilde{O}(1/\epsilon)$ time. We will argue that a change of $\Delta$ in the fractional weight of an edge under $x$ results in at most amortized $O(\Delta/\epsilon^2)$ slot insertions and deletions and $O(\Delta\epsilon)$ bucket merges, deletions, and resets. This implies that the number of slots undergoing updates throughout the update sequence is at most $O(\Delta(x,l_1) /\epsilon^2)$. Since the weight of an edge might change by at most $1$ under any update, the update time of the data structure amortizes to $\tilde{O}(1/\epsilon^2)$.

Whenever a bucket is created, it is either due to a previous bucket splitting in half or due to the merge of two neighboring buckets. Either way, when a bucket is created, it has weight in $[1.5/\epsilon,3/\epsilon]$. Its slots need to undergo $\Omega(1/\epsilon)$ shrinking or stretching in order for the bucket size constraints to be violated. Therefore, there will be at most $O(\epsilon \Delta(x,l_1))$ bucket mergers or splits throughout the updates. 

After each bucket is reset, the algorithm might create one slot with weight at most $1-4\epsilon$. This slot might be deleted due to a very small weight change in $x$. The deletion of the rightmost slot of any bucket might result in only one slot deletion per reset of a bucket. After a reset, all other slots of the bucket are of weight $1-3\epsilon$. Therefore, before it may be reset again due to one of its slots violating weight constraints, its slots must undergo at least an $\epsilon$ change in weight. Hence, the total number of bucket resets prompted by slots violating weight constraints is $O(\Delta(x,l_1)/\epsilon)$, each resulting in $O(1/\epsilon)$ slot updates.

\end{proof}

\subsection{Proof of \Cref{lem:slot-decomposition:matching}}

\label{sec:rounding:matching}

Let $M$ be a matching of $G$ not matching all nodes of $L$, and let $x$ be a fractional matching satisfying the properties of the lemma. For a set of nodes $U \subseteq L \cup R$, let $U[M]$ stand for its subset incident to edges of $M$. Let $L_F = L \setminus L[M]$ and $R_F = R \setminus R[M]$.

Consider a random walk that starts from an arbitrary node of $L_F$ and, in each step, does the following: 1) if it reaches a node of $R_F$, it terminates; 2) if it reaches a node of $R[M]$, it continues toward its neighbor in $M$; and 3) if it reaches a node $l \in L[M]$, it samples a neighbor $r$, each with probability $x(l,r)/x(l)$, and continues toward $r$. We will first argue that such a walk terminates in $O(\log n /\epsilon)$ steps in expectation.

For a node $l \in L[M]$, let $p_t(l)$ stand for the probability that $l$ is the $t+1$-th node of $L$ visited by the walk, and let $q_t(l) = p_t(l)/x(l)$. Denote $Q_t = \max_{l \in L[M]} q_t(l)$. Then, for all $l \in L$ with neighbor $r$ in $M$:

\begin{align*}
q_{t+1}(l) = \frac{p_{t+1}(l)}{x(l)} = \frac{\sum_{l' \in L[M]}\frac{p_t(l') x(l',r)}{x(l')}}{x(l)} \leq \frac{Q_t}{x(l)}\sum_{l' \in L[M]}x(l',r) = \frac{Q_t x(r)}{x(l)}
\end{align*}

As $x(l) \geq 1-\epsilon$ and $x(r) \leq 1-2\epsilon$, we get that $Q_{t+1} \leq (1-\epsilon)Q_t$. By union-bounding across all $l \in L[M]$, we know that the probability of the walk not finishing by $t$ steps is at most $nQ_t$. From the geometric decay of $Q_t$, we can conclude that the walk finishes in $O(\log n /\epsilon)$ steps with probability $1-1/\poly(n)$.

Armed with this claim, it is straightforward to prove the lemma. The algorithm can match any unmatched node of $L$ along the augmenting path described by such a random walk. This means that the algorithm may match an unmatched node of $L$ in $O(\log n / \epsilon)$ time and queries to $\mathcal{O}$. The algorithm needs to do this at most $n$ times under insertions to $L$, and then once whenever a node of $L$ becomes unmatched either through a direct edge deletion or through the deletion of its neighbor in $M$.

%% file: Acknowledgements.tex
\section*{Acknowledgements}
ChatGPT 5.5 provided significant assistance for two of the proof ideas in the paper. 
\begin{itemize}
\item Our analysis of random-augmenting-walk lengths in Lemma \ref{lem:slot-decomposition:matching} is based on an idea from ChatGPT. An earlier analysis was already known, but the idea from ChatGPT led us to a simpler proof with a better dependence on $1/\eps$.
\item The main idea for the conditional lower bound in Appendix \ref{app:integralhardness} came from ChatGPT.
\end{itemize}
ChatGPT also provided technical assistance in pinning down the details of the conditional lower bound in Appendix \ref{ap:matching}. 

Sayan Bhattacharya is funded by the European Union (ERC grant, DYNALP, 101170133).
Views and opinions expressed are however those of the author(s) only and do not necessarily reflect
those of the European Union or the European Research Council Executive Agency. Neither the
European Union nor the granting authority can be held responsible for them.

%% file: appendix-hardness.tex
\section{Dynamic Matching via Dynamic Load Balancing}

\label{ap:matching}

This section is devoted to showing that the existence of an efficient fractional algorithm for fully dynamic load balancing with near-optimal maximum load would imply the existence of an $\epsilon$-approximate fully dynamic bipartite matching algorithm.

For a bipartite graph $G = (L \cup R,E)$, a matching $M \subseteq E$ is a set of vertex-disjoint edges. Let $\mu(G)$ denote the size of a maximum-size matching of $G$. A matching $M$ is said to be $\epsilon$-approximate if $|M| \geq (1-\epsilon)\mu(G)$.

In order to prove the reduction, we will rely on some well-known properties of dynamic matching. First, if we allow randomization, it is known that we may assume that $\mu(G) \geq \epsilon n$.

\begin{proposition}[\cite{Kiss23}]
\label{prop:vertex-sparsification}

Assume there exists a dynamic algorithm that maintains an $\epsilon$-approximate matching of a graph $G$ on $n$ vertices undergoing edge insertions and deletions such that, at all times, $\mu(G) \geq \epsilon n$, with amortized update time $T$. Then there exists a dynamic algorithm maintaining an $\epsilon$-approximate matching of a graph undergoing edge insertions and deletions with high probability in $\tilde{O}(T)$ amortized update time.

\end{proposition}

The following proposition follows from \cite{McGregor05,BhattacharyaKS23Sub}. Note that the statement is also true for general graphs. We refer to the formulation stated by \cite{BehnezhadG24}. 

\begin{proposition}[\cite{BehnezhadG24}]
\label{prop:boosting}

Let $G = (V,E)$ be a graph on $n$ vertices. Assume there exists an algorithm $\mathcal{A}$ that, when queried with a vertex set $U \subseteq L \cup R$ and a parameter $\delta \in (0,1)$ such that $\mu(G[U]) \geq \delta \cdot n$, returns a matching of $G[U]$ of size $\poly(\delta) \cdot \mu(G[U])$. Then, for any $\epsilon \in (0,1)$, there exists an algorithm that, in $\tilde{O}(n)$ time and with $t = O_\epsilon(1)$ adaptive calls, 

$$\mathcal{A}(G,U_1,\delta_1), \dots, \mathcal{A}(G,U_t,\delta_t)$$

returns a matching of $G$ of size at least $\mu(G)-\epsilon n$. Here each $\delta_i$ is a function of $\epsilon$ and is at most $\epsilon/2$.

\end{proposition}

\begin{lemma}
\label{lem:matchingvialoadbalancing}
Assume there exists a dynamic algorithm with amortized update time $T$ that maintains an $\epsilon$-fractional assignment $x$ of an unweighted dynamic load balancing instance undergoing edge insertions and deletions such that, whenever the instance has optimum load $\Gamma$, the assignment $x$ has maximum load at most $(1+\epsilon)\Gamma$.

Then, for $\epsilon \in (0,1)$, there exists a dynamic algorithm that maintains an $\epsilon$-approximate matching of a bipartite graph undergoing edge insertions and deletions in $\tilde O(T)$ amortized update time with high probability against an oblivious adversary.
\end{lemma}

\begin{proof}
By \cref{prop:vertex-sparsification}, it suffices to handle the case in which $\mu(G) \geq \rho n$ at all times, where $\rho=\rho(\epsilon)>0$ is a sufficiently small constant. We run the algorithm in phases of length $\Theta(\epsilon \rho n)$ updates. At the beginning of a phase, we compute a matching of size at least $\mu(G)-(\epsilon\rho/10)n$. During the phase, we output this matching after accounting for edge deletions. Since at most $\Theta(\epsilon\rho n)$ updates occur during the phase and $\mu(G)\geq \rho n$ throughout, choosing the hidden constant in the phase length sufficiently small ensures that the maintained matching remains $(1-\epsilon)$-approximate throughout the entire phase.

It remains to describe how to compute, at the beginning of a phase, a matching of size at least $\mu(G)-(\epsilon\rho/10)n$. For this, we use \cref{prop:boosting}. Thus, it is enough to implement the following primitive. Given $U\subseteq L\cup R$ and $\delta\in(0,1)$ such that $\mu(G[U])\geq \delta n$, where $\delta$ is a constant depending only on $\epsilon$, we must return a matching of size $\poly(\delta)\cdot \mu(G[U])$ in $G[U]$.

Fix such a query $(U,\delta)$. Let
$$
        L_U := L\cap U,\qquad R_U := R\cap U.
$$
Define constants
$$
        \alpha := \delta/4
        \qquad\text{and}\qquad
        \epsilon' \leq \delta/1000.
$$

Here, let $\delta$ stand for the smallest constant \cref{prop:boosting} would ever require to function. Recall that this value is a function of $\epsilon$.

\paragraph{Initialization.}
At initialization, we construct an auxiliary load-balancing instance
$$
        G'=(L\cup L_1 \cup R\cup R_1\cup R_2\cup R_3,\; E\cup E'),
$$
based on $G = (L \cup R,E)$, where $L_1,R_1,R_2,R_3$ are fresh vertex sets, each of size $n$. Initially, $E'$ consists of an arbitrary perfect matching between $L_1$ and $R_1$. $L \cup L_1$ constitutes the set of jobs and $R \cup R_1 \cup R_2 \cup R_3$ constitutes the set of machines.

We maintain, using the assumed dynamic load-balancing algorithm, an $\epsilon'$-fractional assignment $x$ for $G'$. In addition, we maintain a sampled subgraph $H$ of $G'$. Independently for every edge $e\in E\cup E'$, the edge $e$ is included in $H$ with probability
$$
        p_e := \min\{1, C x_e \log n\},
$$
where $C=C(\delta)$ is a sufficiently large constant. Whenever the value $x_e$ changes, we delete $e$ from $H$ and resample it using the new value $x_e$. All random choices are made independently.

\paragraph{Executing a query $\mathcal{A}(G,U,\delta)$.}
We temporarily modify $G'$ as follows. First, choose a subset $L_1'\subseteq L_1$ of size $|R\setminus U|$. Delete from $E'$ all edges incident to vertices of $L_1'$, and insert an arbitrary perfect matching between $L_1'$ and $R\setminus U$. Second, insert an arbitrary perfect matching between $L\setminus U$ and a set $R_2'\subseteq R_2$ of size $|L\setminus U|$.

Third, choose a set $R_3'\subseteq R_3$ of size $|R_3'| = |L_U|-\alpha n$. If this number is negative, then $\mu(G[U])\leq |L_U|<\alpha n<\delta n$, contradicting the promise $\mu(G[U])\geq \delta n$. Hence $R_3'$ is well-defined. For every pair $(u,v)\in L_U\times R_3'$, insert the edge $(u,v)$ independently with probability
$$
        q := C\delta^{-1}\frac{\log n}{n},
$$
where $C$ is a sufficiently large absolute constant.

After these temporary updates have been processed by the load-balancing algorithm, we inspect the sampled graph $H$ and compute a maximal matching in $H[U]$. Let this matching be $M_U$. We output $M_U$ as the answer to the query, and then reverse all temporary updates made above. We now prove that $M_U$ has the required size with high probability.

\begin{claim}
\label{cl:opt-load-one}
After the temporary updates for the query $(U,\delta)$, the optimum load of $G'$ is $1$ with high probability.
\end{claim}

\begin{proof}
It is enough to exhibit an integral assignment of every left vertex to a distinct right vertex. The vertices of $L_1'$ are assigned to $R\setminus U$ along the inserted perfect matching. The vertices of $L_1\setminus L_1'$ are assigned to $R_1$ along their original matching edges. The vertices of $L\setminus U$ are assigned to $R_2'$ along the inserted perfect matching.

It remains to assign the vertices of $L_U$. Let $M_U^*$ be a maximum matching of $G[U]$. By assumption,
$$
        |M_U^*|=\mu(G[U])\geq \delta n.
$$
Let $L^*\subseteq L_U$ be the vertices of $L_U$ matched by $M_U^*$. We assign the vertices of $L^*$ to $R_U$ along the matching $M_U^*$. It remains to assign $A := L_U\setminus L^*$ to $R_3'$. Since $|L^*|\geq \delta n$, we have
$$
        |A|
        =
        |L_U|-|L^*|
        \leq
        |L_U|-\delta n.
$$
On the other hand,
$$
        |R_3'|
        =
        |L_U|-\delta n/4.
$$
Thus
$$
        |R_3'|-|A|
        \geq
        3\delta n/4.
$$

We use the following standard fact. Let $A,B$ be two vertex sets with $|B|-|A|\geq \eta n$, where $\eta>0$ is a constant, and expose a random bipartite graph between $A$ and $B$ by including every edge independently with probability at least $C\eta^{-1}\log n/n$. For sufficiently large $C$, this random bipartite graph contains a matching saturating $A$ with probability at least $1-n^{-10}$.

Applying this fact with $B=R_3'$ and $\eta=3\delta/4$, the random graph between $A$ and $R_3'$ contains a matching saturating $A$ with high probability. Assigning the vertices of $A$ along this matching completes an integral assignment of load $1$. Therefore the optimum load of the temporary instance is $1$ with high probability.
\end{proof}

By \cref{cl:opt-load-one}, during the query the assignment $x$ has maximum load at most $1+\epsilon'$, and every left vertex has total assigned mass at least $1-\epsilon'$.

\begin{claim}
\label{cl:fractional-matching-inside-U}
The restriction of $x$ to edges of $G[U]$ has total weight at least $\delta n/8$.
\end{claim}

\begin{proof}
Let $m:=|L_U|$ and $r:=|R\setminus U|$. The only neighbors of the vertices in $L_1'$ are the vertices of $R\setminus U$, matched one-to-one by the temporary perfect matching. Since every left vertex is assigned mass at least $1-\epsilon'$, the vertices of $L_1'$ send total mass at least $(1-\epsilon')r$ to $R\setminus U$.

The total capacity, under the assignment $x$, of the vertices in $(R\setminus U)\cup R_3'$ is at most $(r+|R_3'|)(1+\epsilon')$. Therefore the total mass that vertices of $L_U$ can send to $(R\setminus U)\cup R_3'$ is at most
$$
        (r+|R_3'|)(1+\epsilon')-(1-\epsilon')r
        =
        |R_3'|(1+\epsilon')+2\epsilon' r.
$$
Using $|R_3'|=m-\delta n/4$ and $r\leq n$, this is at most
$$
        m-\delta n/4+\epsilon' m+2\epsilon'n
        \leq
        m-\delta n/4+3\epsilon'n.
$$
Since every vertex of $L_U$ is assigned mass at least $1-\epsilon'$, the total mass sent out of $L_U$ is at least
$$
        (1-\epsilon')m.
$$
Hence the mass that $L_U$ sends to $R_U$ along edges of $G[U]$ is at least
$$
        (1-\epsilon')m-\left(m-\delta n/4+3\epsilon'n\right)
        =
        \delta n/4-\epsilon'm-3\epsilon'n.
$$
Because $m\leq n$ and $\epsilon'\leq \delta/1000$, this quantity is at least
$$
        \delta n/4-4\epsilon'n
        \geq
        \delta n/8.
$$
Thus $x$ restricted to the original edges of $G[U]$ has total weight at least $\delta n/8$.
\end{proof}

Since the maximum load of $x$ is at most $1+\epsilon'$, the vector
$$
        y := \frac{x}{1+\epsilon'}
$$
is a feasible fractional matching: every left and right vertex is incident to total $y$-mass at most $1$. By \cref{cl:fractional-matching-inside-U}, the total $y$-mass on edges of $G[U]$ is at least
$$
        \frac{\delta n}{8(1+\epsilon')}
        \geq
        \delta n/16.
$$
Denote this quantity by $W$. Thus $G[U]$ contains a fractional matching of value $W\geq \delta n/16$ supported on the edges whose sampled probabilities are proportional to their $x$-values.

\begin{claim}
\label{cl:H-size}
With high probability, throughout the updates the sampled graph $H$ has $\tilde O(n)$ edges.
\end{claim}

\begin{proof}
At any time, the total $x$-mass over all edges is at most the number of left vertices times the maximum assigned mass per left vertex. There are $|L|+|L_1|\leq 2n$ left vertices, and each sends total mass at most $1+\epsilon'$ after truncating any excess mass if necessary. Hence
$$
        \sum_{e} x_e \leq 2(1+\epsilon')n.
$$
Therefore
$$
        \mathbb{E}[|E(H)|]
        =
        \sum_e \min\{1,Cx_e\log n\}
        \leq
        C\log n \sum_e x_e
        =
        O(n\log n).
$$
By a Chernoff bound, $|E(H)|=\tilde O(n)$ with probability at least $1-n^{-10}$ at every fixed time. Taking a union bound over the polynomially many times at which the maintained fractional assignment changes throughout the updates gives the claim.
\end{proof}

\begin{claim}
\label{cl:H-approx}
With high probability, $H[U]$ contains a matching of size at least $\Omega(\delta n)$. Consequently, the maximal matching $M_U$ computed in $H[U]$ has size at least $\Omega(\delta n)$.
\end{claim}

\begin{proof}
Let $F$ be the set of original edges of $G[U]$, and let $y$ be the fractional matching from above supported on $F$ with total value $W=\sum_{e\in F} y_e \geq \delta n/16$. For every $e\in F$, $x_e=(1+\epsilon')y_e$, and therefore $p_e=\min\{1,Cx_e\log n\}\geq\min\{1,Cy_e\log n\}$.

Let $F_H:=F\cap E(H)$. First, we lower bound $|F_H|$. We have
$$
        \mathbb{E}[|F_H|]
        =
        \sum_{e\in F} p_e
        \geq
        \sum_{e\in F}\min\{1,Cy_e\log n\}.
$$
Since $0\leq y_e\leq 1$, this expectation is at least
$$
        \Omega(CW\log n).
$$
Choosing $C$ sufficiently large and applying a Chernoff bound, we get
$$
        |F_H| \geq c_1 W\log n
$$
with probability at least $1-n^{-10}$, for an absolute constant $c_1>0$.

Second, we upper bound the maximum degree of $H[U]$. Since $y$ is a fractional matching, for every vertex $v$,
$$
        \sum_{e\ni v} y_e \leq 1.
$$
Thus the expected degree of $v$ in $F_H$ is at most $C(1+\epsilon')\log n$. Again by Chernoff and a union bound over all vertices,
$$
        \Delta(F_H) \leq c_2 C\log n
$$
with probability at least $1-n^{-10}$, for an absolute constant $c_2>0$.

On the intersection of these two high-probability events, a greedy maximal matching in $F_H$ has size at least
$$
        \frac{|F_H|}{2\Delta(F_H)}
        \geq
        \frac{c_1 W\log n}{2c_2C\log n}
        =
        \Omega(W)
        =
        \Omega(\delta n).
$$
Since $\mu(G[U])\leq n$, this is also at least $\poly(\delta)\cdot \mu(G[U])$. More explicitly, by increasing the constants above if necessary, we may guarantee
$$
        |M_U|\geq c\delta\cdot \mu(G[U])
$$
for some absolute constant $c>0$. This proves the claim.
\end{proof}

By \cref{cl:H-approx}, each call required by \cref{prop:boosting} can be answered with a matching of size $\poly(\delta)\mu(G[U])$. The number of calls is $O_\epsilon(1)$, and each call performs $\tilde O(n)$ temporary edge updates to the auxiliary load-balancing instance and then computes a maximal matching in a graph of size $\tilde O(n)$ by \cref{cl:H-size}. Since phases have length $\Theta_\epsilon(n)$, the amortized overhead per original update is $\tilde O(T)$.
\end{proof}

\section{$2-\epsilon$-Approximation Hardness for Full Dynamic Load Balancing}

\label{app:integralhardness}

In this section, we show that the existence of a better-than-$2$-approximate algorithm with $o(n)$ update time for the unweighted fully dynamic load balancing problem would violate the OMv conjecture \cite{henzinger2015unifying}.

\paragraph{The hardness assumption.}
We use the following standard consequence of the OMv conjecture. In the OuMv problem, after preprocessing a Boolean matrix \(A\in\{0,1\}^{n\times n}\), an algorithm receives \(n\) online pairs of Boolean vectors \((u^t,v^t)\in\{0,1\}^n\times\{0,1\}^n\), and after receiving \((u^t,v^t)\) must output
$$
(u^t)^\top A v^t=\bigvee_{i,j\in[n]}(u^t_i\wedge A_{ij}\wedge v^t_j)
$$
before seeing the next pair. Under the OMv conjecture, for every constant \(\gamma>0\), OuMv cannot be solved with polynomial preprocessing and total online time \(O(n^{3-\gamma})\).

\begin{theorem}
Assume the OMv conjecture. For every pair of constants \(\varepsilon,\delta>0\), there is no fully dynamic \((2-\varepsilon)\)-approximation algorithm for bipartite load balancing with polynomial preprocessing, amortized update time \(O(N^{1-\delta})\), and query time \(O(N^{2-\delta})\), where \(N=|J|+|M|\). In particular, if the maintained assignment value is explicitly available after every update, then there is no such algorithm with amortized update time \(O(N^{1-\delta})\).
\end{theorem}

\begin{proof}
Suppose, toward contradiction, that such a dynamic algorithm \(\mathcal D\) exists. We show how to solve OuMv in truly subcubic time.

Let \(A\in\{0,1\}^{n\times n}\) be the OuMv matrix. From \(A\), construct a bipartite graph \(G=(J\cup M,E)\). The job side is
$$
J=\{x,y\}\cup\{a_i:i\in[n]\}\cup\{b_j:j\in[n]\},
$$
and the machine side is
$$
M=\{p\}\cup\{r_i:i\in[n]\}\cup\{c_j:j\in[n]\}\cup\{f_j:j\in[n]\}.
$$
Thus \(N=|J|+|M|=5n+3=\Theta(n)\).

The static edges are
$$
(y,p),\ (x,p),\quad (a_i,r_i)\ \forall i\in[n],\quad (b_j,c_j)\ \forall j\in[n],\quad (a_i,c_j)\ \text{iff}\ A_{ij}=1.
$$
For the current online pair \((u,v)\), the dynamic edges are
$$
(x,r_i)\ \text{iff}\ u_i=1,\quad\text{and}\quad (b_j,f_j)\ \text{iff}\ v_j=1.
$$
When the online pair changes from \((u,v)\) to \((u',v')\), we update exactly the edges \((x,r_i)\) for which \(u_i\neq u'_i\) and the edges \((b_j,f_j)\) for which \(v_j\neq v'_j\). Hence each OuMv round requires at most \(2n\) edge updates. Fix a pair \((u,v)\).

First suppose \(u^\top A v=1\). Then there exist \(i,j\in[n]\) such that \(u_i=1\), \(A_{ij}=1\), and \(v_j=1\). Consider the following assignment:
$$
y\mapsto p,\quad x\mapsto r_i,\quad a_i\mapsto c_j,\quad b_j\mapsto f_j,\quad a_{i'}\mapsto r_{i'}\ \forall i'\neq i,\quad b_{j'}\mapsto c_{j'}\ \forall j'\neq j.
$$
Every used edge exists: \((x,r_i)\) exists because \(u_i=1\), \((a_i,c_j)\) exists because \(A_{ij}=1\), and \((b_j,f_j)\) exists because \(v_j=1\). No machine receives more than one job. Therefore \(\operatorname{OPT}(G)\le 1\), and since every job has unit demand, \(\operatorname{OPT}(G)=1\).

Now suppose \(u^\top A v=0\). We claim that no load-\(1\) assignment exists. A load-\(1\) assignment is a matching that saturates all jobs. Since \(y\) is adjacent only to \(p\), such an assignment must assign \(y\) to \(p\). Therefore \(x\) cannot be assigned to \(p\), so \(x\) must be assigned to some \(r_i\). The edge \((x,r_i)\) exists only if \(u_i=1\). Since \(r_i\) is occupied by \(x\), job \(a_i\) cannot be assigned to \(r_i\), so \(a_i\) must be assigned to some \(c_j\). The edge \((a_i,c_j)\) exists only if \(A_{ij}=1\). Since \(c_j\) is occupied by \(a_i\), job \(b_j\) cannot be assigned to \(c_j\), so \(b_j\) must be assigned to \(f_j\). The edge \((b_j,f_j)\) exists only if \(v_j=1\). Hence \(u_i=1\), \(A_{ij}=1\), and \(v_j=1\), contradicting \(u^\top A v=0\). Thus no load-\(1\) assignment exists.

On the other hand, a load-\(2\) assignment always exists:
$$
x\mapsto p,\quad y\mapsto p,\quad a_i\mapsto r_i\ \forall i\in[n],\quad b_j\mapsto c_j\ \forall j\in[n].
$$
Therefore, when \(u^\top A v=0\), we have \(\operatorname{OPT}(G)=2\). Altogether,
$$
u^\top A v=1\Rightarrow \operatorname{OPT}(G)=1,\qquad u^\top A v=0\Rightarrow \operatorname{OPT}(G)=2.
$$

Now run \(\mathcal D\) on this dynamic graph. After the updates for the current pair \((u,v)\), query the value of the maintained assignment. If \(u^\top A v=1\), then \(\operatorname{OPT}(G)=1\), so a \((2-\varepsilon)\)-approximation must have load at most \(2-\varepsilon<2\). Since loads are integral, the maintained load is exactly \(1\). If \(u^\top A v=0\), then \(\operatorname{OPT}(G)=2\), so every feasible assignment has load at least \(2\). Therefore the maintained load is \(1\) if and only if \(u^\top A v=1\). This solves the current OuMv query.

It remains to bound the running time. The graph has \(N=\Theta(n)\) vertices. Over all \(n\) online pairs, there are at most \(2n^2\) edge updates and \(n\) queries. Therefore the total online time of the resulting OuMv algorithm is
$$
O(n^2\cdot N^{1-\delta}+n\cdot N^{2-\delta})=O(n^2\cdot n^{1-\delta}+n\cdot n^{2-\delta})=O(n^{3-\delta}).
$$
The preprocessing of \(\mathcal D\) is polynomial in the initial graph size and hence polynomial in \(n\), so it is allowed by the polynomial-preprocessing version of OuMv hardness. We have obtained an \(O(n^{3-\delta})\)-time online algorithm for OuMv after polynomial preprocessing, contradicting the OMv conjecture. This proves the theorem.
\end{proof}

\paragraph{Remark.}
The reduction is tight for the argument used above: it creates a \(1\) versus \(2\) optimum gap. Hence it rules out every \((2-\varepsilon)\)-approximation, but it does not by itself rule out \(2\)-approximations or larger constant-factor approximations.